\documentclass[sigconf]{acmart}

\copyrightyear{2021}
\acmYear{2021}
\setcopyright{acmcopyright}
\acmConference[CCS '21] {Proceedings of the 2021 ACM SIGSAC Conference on Computer and Communications Security}{November 15--19, 2021}{Virtual Event, Republic of Korea.}
\acmBooktitle{Proceedings of the 2021 ACM SIGSAC Conference on Computer and Communications Security (CCS '21), November 15--19, 2021, Virtual Event, Republic of Korea}
\acmPrice{15.00}
\acmISBN{978-1-4503-8454-4/21/11}
\acmDOI{10.1145/3460120.3484752}

\begin{document}
\fancyhead{}

\title{How Does Blockchain Security Dictate Blockchain Implementation? }

\author{Andrew Lewis-Pye}
\affiliation{%
  \institution{London School of Economics}
  }
\email{a.lewis7@lse.ac.uk}

\author{Tim Roughgarden}
\affiliation{%
  \institution{Columbia University}
  }
\email{tim.roughgarden@gmail.com}

%
%



\begin{abstract}
   Blockchain protocols come with a variety of security guarantees.  For example, BFT-inspired protocols such as Algorand\footnote{For an exposition of Algorand that explains how to achieve security in the partially synchronous setting, see \cite{chen2018algorand}.} tend to be secure in the partially synchronous setting, while longest chain protocols like Bitcoin will normally require stronger synchronicity to be secure. 
Another fundamental distinction,  directly relevant to scalability
solutions such as sharding, is whether or not a single untrusted user
is able to point to \emph{certificates}, which provide
incontrovertible proof of block confirmation. Algorand produces such
certificates, while Bitcoin does not.  Are these properties
accidental? Or are they inherent consequences of the paradigm of
protocol design?  Our aim in this paper is to understand what,
fundamentally, governs the nature of security  for permissionless
blockchain protocols. Using the framework developed in \cite{lewis2021byzantine},
we prove general results showing that these questions relate directly
to properties of the user selection process, i.e.\ the method (such as
proof-of-work or proof-of-stake) which is used to select users with
the task of updating state. Our results suffice to establish, for
example,  that the production of certificates is impossible for
proof-of-work protocols, but is automatic for standard forms of
proof-of-stake protocols. As a byproduct of our work, we also define a number of security notions and identify the equivalences and inequivalences among them. 
\end{abstract}



\maketitle

\section{Introduction} \label{intro} 

\paragraph{Paradigms for blockchain protocol design.}
In the wake of Bitcoin \cite{nakamoto2008bitcoin}, thousands of cryptocurrencies have flooded the market. While many of these currencies use only slight modifications of the Bitcoin protocol, there are also a range of cryptocurrencies taking radically different design approaches. 
Two informal distinctions are between: 

\begin{enumerate}
\item Proof-of-stake (PoS)/proof-of-work (PoW). In a PoW protocol, users are selected and given the task of updating state, with the probability any particular user is chosen being proportional to their (relevant) computational power. In PoS protocols, users are selected with probability depending on their stake (owned currency).  
\item BFT\footnote{The acronym BFT stands for `Byzantine-Fault-Tolerant'.}/longest-chain. As well as being a PoW protocol, Bitcoin is the best known example of a longest chain protocol.  This means that forks may occur in the blockchain, but that honest miners will
build on the longest chain. In a BFT protocol, on the other hand, users are selected and asked to carry
out a consensus protocol designed for the permissioned setting. So, roughly, longest chain protocols are those which are derived from Bitcoin, while BFT protocols are derived from protocols designed in the permissioned setting. Algorand \cite{chen2016algorand} is a well known example of a BFT protocol.
\end{enumerate}

\paragraph{A formal framework for comparing design paradigms~\cite{lewis2021byzantine}.}
While informal, these distinctions are more than aesthetic. For example, BFT protocols like Algorand will tend to give security guarantees that hold under significantly weaker network connectivity assumptions than are required to give security for protocols like Bitcoin. By developing an appropriate formal framework,  it can then be shown \cite{lewis2021byzantine}  that these differences in security are a \emph{necessary} consequence of the paradigm of protocol design: The fact that Bitcoin is a PoW protocol means that it cannot offer the same flavour of security guarantees as Algorand. A framework of this kind was developed in \cite{lewis2021byzantine}, according to which permissionless \footnote{In the distributed computing literature, consensus protocols have traditionally been studied in a setting where all participants  are known to each other from the start of the protocol execution. In the parlance of the blockchain literature, this is referred to as the  \emph{permissioned} setting. What differentiates  Bitcoin  \cite{nakamoto2008bitcoin} from these previously studied protocols is that it operates in a \emph{permissionless} setting, i.e.\ it is a protocol for establishing consensus over an unknown network of participants that anybody can join, with as many identities as they like in any role. } protocols run relative to a \emph{resource pool}. This resource pool specifies a balance for each user over the duration of the protocol execution (such as hashrate or stake), which may be used in determining which users are permitted to update state. Within this framework, the idea that protocols like Bitcoin require stronger connectivity assumptions for security can be formalised as a theorem asserting that \emph{adaptive} protocols cannot be \emph{partition secure} -- these terms apply to permissionless blockchain protocols and will be defined formally later on, but, roughly, they can be summed up as follows:
\begin{itemize}
\item   \emph{Liveness} and \emph{security} are  defined in terms of a notion of \emph{confirmation} for blocks. 
A protocol is live if the number of confirmed blocks can be relied on to increase during extended intervals of time during which message delivery is reliable. A  protocol is secure if rollback on confirmed blocks is unlikely. 
\item Bitcoin being adaptive means that it remains live in the face of an unpredictable size of resource pool (unpredictable levels of mining).  
\item A protocol is partition secure if it is secure in the \emph{partially synchronous} setting, i.e.\ if the rollback of confirmed blocks remains unlikely even in the face of potentially unbounded network partitions.  The partially synchronous setting will be further explained and formally defined in Section \ref{framework}.
\end{itemize}


\paragraph{This paper: certificates.}
The way in which Algorand and other BFT protocols achieve partition security is also worthy of note. For all such protocols, protection against unbounded network partitions is provided through the production of \emph{certificates}: These are sets of broadcast messages whose very existence suffices to establish block confirmation and which cannot be produced by a (suitably bounded) adversary given the entire duration of the execution of the protocol.  Bitcoin does not produce certificates, because the existence of a certain chain does generally not prove that it is the longest chain -- a user will only believe that a certain chain is the longest chain until presented with a longer (possibly incompatible) chain.  Algorand does produce certificates, on the other hand, because the very existence of a valid chain, together with appropriate committee signatures for all the blocks in the chain, suffices to guarantee (beyond a reasonable doubt) that the blocks in that chain are confirmed. 
 We will formally define what it means for a protocol to produce certificates in Section \ref{pss}.

 The production of certificates is also functionally useful, beyond providing security against network partitions. The production of certificates means, for example,  that a  single untrusted user is able to convince another user of block confirmation (by relaying an appropriate certificate), and this is potentially very useful in the context of sharding. If a user wishes to learn the state of a blockchain they were not previously monitoring, then it is no longer necessary to perform an onboarding process in which one samples the opinions of users until such a point that it is likely that at least one of them was `honest' -- one simply requests a certificate proving confirmation for a recently timestamped block.\footnote{Such techniques can  avoid the need to store block hashes in a  sharding `main chain', and the information withholding attacks that come with those approaches.}

\subsection{Overview of results.}
The goal of this paper is to rigorously investigate to what extent
today's protocols ``have to look the way they are'' given the security
guarantees they achieve. 
Such formal analyses are relevant to the broader research community
for several reasons, including:
(i) accurate intuitions of the community (e.g., that there's
fundamentally only one way to achieve certain properties) can be
formally validated, with the necessary assumptions clearly spelled out;
(ii) inaccurate intuitions can be exposed as such;
(iii) unexplored areas of the protocol design space can naturally rise
to the surface (e.g., Section~\ref{ss:recalibration}); and (iv) new
definitions (e.g., certificates) can enhance our language for crisply
describing and comparing competing solutions (both present and future).
In this paper, we prove three main results, which each address this issue in a different setting. 


\vspace{0.2cm} \noindent \textbf{The partially synchronous setting.}
The first key question is:

\begin{enumerate} 
\item[Q1.] Are certificates fundamental to partition security, or an artifact of Algorand's specific implementation?  That is, are certificates the \emph{only} way for permissionless blockchain protocols to achieve security in the partially synchronous setting?
\end{enumerate} 

\noindent Our first main result, proved in the context of the
framework of \cite{lewis2021byzantine}, gives an affirmative response to Q1. Of course, all terms will be explained and formally defined in later sections. 

\vspace{0.2cm}  THEOREM 3.3. \emph{If a permissionless blockchain protocol is secure in the partially synchronous setting, then it produces certificates.}
\vspace{0.2cm} 

\noindent Since it will be easily observed that the production of
certificates implies security, Theorem \ref{strongsame} shows that, in
the partially synchronous setting, the production of certificates is
actually \emph{equivalent} to security. 

\vspace{0.2cm} \noindent \textbf{The synchronous setting.} What about Bitcoin?  While
Bitcoin does not satisfy the conditions of Theorem \ref{strongsame},
it clearly has some non-trivial security.  The standard formalisation
in the literature
\cite{ren2019analysis,garay2018bitcoin} is that
Bitcoin is secure in the \emph{synchronous setting}, for which
there is an upper bound on message delivery time.\footnote{The
  synchronous setting will be further explained and formally defined
  in Section \ref{framework}.} Even working in the synchronous
setting, though, it is clear that Bitcoin does not produce
certificates.  Again, we are led to ask whether this is a necessary
consequence of the paradigm of protocol design:

\begin{enumerate} 
\item[Q2.] Could there be a Bitcoin-like protocol that, at least in the synchronous setting, has as strong a security guarantee in terms of the production of certificates as BFT-type protocols do in the partially synchronous setting? 
\end{enumerate}

\noindent 
The answer depends on key features of the resource pool -- recall that the resource pool specifies a balance for each user over the duration of the protocol execution, such as hashrate or stake. The crucial
distinction here is between scenarios in which the size of the
resource pool is known (e.g. PoS), and scenarios where the size of the
resource pool is unknown (e.g. PoW). As per the framework in
\cite{lewis2021byzantine}, we will refer to these as the \emph{sized} and
\emph{unsized} settings, respectively -- formal definitions will be
given in Section \ref{cs}.  As alluded to above, we define a protocol
to be adaptive if it is is live in the unsized setting, and it was shown 
in \cite{lewis2021byzantine} that adaptive protocols cannot be secure in the
partially synchronous setting. 

\vspace{0.2cm} \noindent \textbf{The synchronous and unsized setting.}  The term ``non-trivial adversary'',
which is used in Theorem \ref{t2} below, will be defined in Section
\ref{cs} so as to formalise the idea that the adversary may have at
least a certain minimum resource balance throughout the
execution. With these basic definitions in place, we can then give a
negative answer to Q2.


\vspace{0.2cm} 
THEOREM 5.1 
 \emph{Consider the  synchronous and unsized setting. If a permissionless blockchain protocol is live then, in the presence of a non-trivial adversary,  it does not produce  certificates.}
 \vspace{0.2cm} 
 
 So, while Theorem \ref{strongsame} showed that the production of certificates is \emph{necessary} in the partially synchronous setting, Theorem \ref{t2} shows that the production of certificates isn't \emph{possible} in the unsized setting (in which PoW protocols like Bitcoin operate). Following on
 from our previous discussion regarding the relevance of certificates
 to sharding, one direct application of this result is that it rules
 out certain approaches to sharding for PoW protocols.  
 
\vspace{0.2cm} \noindent \textbf{The synchronous and sized setting.}   In the sized
 setting (such as for PoS protocols), though, it is certainly
 \emph{possible} for protocols to produce certificates. It therefore
 becomes a natural question to ask how far we can push this:
 
\begin{enumerate} 
\item[Q3.] Does the production of certificates come down purely to properties of the process of user selection? Is it simply a matter of whether one is in the sized or unsized setting?  
\end{enumerate} 

\noindent Our final theorem gives a form of positive response to Q3. We state an informal version of the theorem below. A formal version will be given in Section \ref{cs}.

\vspace{0.2cm} THEOREM 5.6 (INFORMAL VERSION). 
\emph{Consider the synchronous and sized setting, and suppose a permissionless blockchain protocol is of `standard form'. Then there exists a `recalibration' of the protocol which produces certificates. }
\vspace{0.2cm} 

\noindent Theorem \ref{t3} says, in particular, that all `standard'
PoS protocols can be tweaked to get the strongest possible security
guarantee, since being of `standard form' will entail satisfaction of
a number of conditions that are normal for such protocols. Roughly
speaking, one protocol will be considered to be a recalibration of
another if running the former just involves running the latter for a
computable transformation of the input parameters and/or using a
different notion of block confirmation.  The example of Snow
White~\cite{bentov2016snow} may be instructive here (for the purposes of this example, the particulars of the Snow White protocol are not important -- all that matters is that,  at a high
level, Snow White might be seen as a PoS version of Bitcoin, but with the fundamental differences that it operates in the sized setting, and that blocks have non-manipulable timestamps). Snow White is a
PoS longest chain protocol, and it is not difficult to see that, with
the standard notion of confirmation, it does not produce certificates
-- an adversary can produce chains of blocks which are \emph{not}
confirmed, but which \emph{would be} considered confirmed in the
absence of other blocks which have been broadcast. So whether a block
is confirmed depends on the whole set of broadcast messages. On the
other hand, it is also not difficult to adjust the notion of
confirmation so that Snow White \emph{does} produce certificates. An
example would be to consider a block confirmed when it belongs to a
long chain of sufficient \emph{density} (meaning that it has members
corresponding to most possible timeslots) that it could not likely be
produced by a (sufficiently bounded) adversary.  We will see further
examples like this explained in greater depth in
Section~\ref{cs}. Theorem \ref{t3} implies much more generally that
PoS protocols can always be modified so as to produce certificates in
this way.

\paragraph{The punchline.}
Whether or not a permissionless blockchain protocol produces certificates comes down essentially
to whether one is working in the sized or unsized setting (e.g.\
whether the protocol is PoS or PoW). This follows from the following
results that we described above: 
\begin{enumerate}
\item[(i)] According to the results of \cite{lewis2021byzantine}, only protocols which work in the sized setting can be secure in the partially synchronous setting. According to Theorem \ref{strongsame}, all such protocols produce certificates. 
\item[(ii)] Theorem \ref{t2} tells us that, in the synchronous and unsized setting, protocols cannot produce certificates. 
\item[(iii)] Theorem \ref{t3}  tells us that all  \emph{standard} protocols in the sized and synchronous setting can be recalibrated to produce certificates.  
\end{enumerate}

  \subsection{Related work}
  
  There are a variety of papers from the distributed computing literature that analyse settings somewhere between the permissioned and permissionless settings as considered here. In \cite{okun2005distributed}, for example, Okun considered a setting which a fixed number of processors communicate by private channels, where each processor may or may not have a unique identifier, and where processors may or may not be `port aware', i.e.\ be able to tell which channel a message arrives from.  A number of papers \cite{cavin2004consensus,alchieri2008byzantine} have also considered the problem of reaching consensus amongst unknown participants (CUP). In the framework considered in those papers, the number and the identifiers of other participants may be unknown from the start of the protocol execution. A fundamental difference with the permissionless setting considered here is that, in the CUP framework,  all participants have a unique identifier and the adversary is unable to obtain additional identifiers to be able
to launch a sybil attack against the system, i.e.\ the number of identifiers controlled by the adversary is bounded.

  The Bitcoin protocol was first described in 2008 \cite{nakamoto2008bitcoin}. Since then, a number of papers \cite{garay2018bitcoin,WHGSW16} have developed frameworks for the analysis of  Bitcoin in which oracles are introduced for modelling PoW. A more general form of oracle is required for modelling PoS and other forms of permissionless protocol, however. In \cite{lewis2021byzantine} a framework was introduced that described a generalised form for  such oracles.  We use that framework in this paper, but also develop that framework in Sections \ref{blockstructure}, \ref{probs},  \ref{defl},  \ref{sec} and \ref{recal}  to be appropriate specifically for the analysis of blockchain protocols.

\section{The Framework}   \label{framework} 

We work within the framework of \cite{lewis2021byzantine}.  While we describe the framework in its entirety here, we refer the reader to the the original paper for further examples and explanations of the framework set-up. Within Section \ref{framework}, it is the definitions of Sections 2.4, 2.5, 2.7 and 2.8 that are new to this paper (all definitions of Sections \ref{pss}, \ref{sh} and \ref{cs} are also new to this paper). 

Most of this section can be briefly summed up as follows -- all undefined terms in the below will be formalised and defined in later subsections. 
\begin{itemize} 
\item Protocols are executed by an unknown number of users, each of which is formalised as a deterministic processor that controls a set of public keys. 
\item Processors have the ability to \emph{broadcast} messages to all other processors. 
The duration of the execution, however, may be divided into \emph{synchronous} or \emph{asynchronous} intervals. During asynchronous intervals, an \emph{adversary} can tamper with message delivery as they choose. During synchronous intervals there is a given upper bound on message delivery time.     We then distinguish two
\emph{synchronicity settings}. In the \emph{synchronous} setting it is
assumed that there are no asynchronous intervals,
while in the \emph{partially synchronous} setting there may be
unpredictably long asynchronous intervals. 
\item Amongst all broadcast messages, there is a distinguished set referred
to as \emph{blocks}, and one block which is referred to as the
\emph{genesis block}.  Unless it is the genesis block, each block $B$
has a unique \emph{parent} block.
\item To blackbox the process of user selection, whereby certain users are selected and given the task of updating state,  \cite{lewis2021byzantine} introduces two new notions: (1) Each public key is considered to have a certain \emph{resource balance}, which may vary over the execution, and; (2) The protocol will also be run relative to a \emph{permitter oracle}, which may respond to this resource balance.  For a PoW protocol like Bitcoin, the resource balance of each public key will be their (relevant)
computational power at the given timeslot.
\item It is the permitter oracle which then gives permission to broadcast messages updating state. To model Bitcoin, for example, we sometimes have the permitter allow another user to broadcast a new block, with the probability this happens for each user being proportional to their resource balance. 
\item Liveness and security are defined in terms of a notion of \emph{confirmation} for blocks. 
Roughly, a protocol is live if the number of confirmed blocks can be relied on to increase during extended intervals of time during which message delivery is reliable. A protocol is secure if rollback on confirmed blocks is unlikely.
\end{itemize}

\subsection{The computational model}  \label{cm} 
\noindent \textbf{Overview.} There are a number of papers analysing Bitcoin \cite{garay2018bitcoin,WHGSW16} that take the approach of working within the language of the UC framework of Canetti \cite{canetti2001universally}. Our position is that this provides a substantial barrier to entry for researchers in blockchain who do not have a strong background in security, and that the power of the UC framework remains essentially unused in the subsequent analysis. Instead, we use a very simple  computational model, which is designed to be as similar as possible to standard models from distributed computing (e.g. \cite{DLS88}),  while also being adapted to deal with the permissionless setting. We thus consider an information theoretic model in which processors are simply specified by state transition diagrams. A \emph{permitter oracle} is introduced as a generalisation of the random oracle functionality in the Bitcoin Backbone paper \cite{garay2018bitcoin}:  It is the permitter oracle's role to grant \emph{permissions} to broadcast messages.  The duration of the execution is divided into timeslots. Each processor enters each timeslot $t$ in a given \emph{state} $x$, which determines the instructions for the processor in that timeslot -- those instructions may involve broadcasting messages, as well as sending \emph{requests} to the permitter oracle. The state $x'$ of the processor at the next timeslot is determined by the state $x$, together with the messages and permissions received at $t$.

Since we focus on impossibility results, we simplify the presentation by making the assumption that we are always working in the \emph{authenticated} setting, in which processors have access to public/private key pairs. This assumption is made purely for the sake of simplicity, and the results of the paper do not depend upon it. 

\vspace{0.2cm} 
\noindent \textbf{Formal description.}  For a list of commonly used variables and terms, see Table \ref{terms} in the appendix. We consider a finite\footnote{In \cite{lewis2021byzantine}, a potentially infinite number of processors were allowed, but each processor was given a single public key (identifier). Here, we will find it convenient to consider instead a finite number of processors, each of which may control an unbounded number of public keys.} system of \emph{processors}. Each processor $p$ is specified by a state transition diagram, for which the number of states may be infinite. 
Amongst the states of a processor are a non-empty set of possible \emph{initial states}. The \emph{inputs} to $p$ determine which initial state it starts in.  If a variable is specified as an input to $p$, then we refer to it as \emph{determined} for $p$, referring to the variable as \emph{undetermined}  for $p$ otherwise.
If a variable is determined/undetermined for all $p$, we simply refer to it as determined/undetermined. 
Amongst the inputs to $p$ is an infinite set $\mathcal{U}_p$ of public keys,  which are specific to $p$ in the sense that if $\mathtt{U}\in \mathcal{U}_p$ and $\mathtt{U}'\in \mathcal{U}_{p'}$ then $\mathtt{U}\neq \mathtt{U}' $ when $p\neq p'$.  A principal difference between the permissionless setting (as considered here) and the permissioned setting (as studied in classical distributed computing) is that, in the permissionless setting, the number of processors is undetermined, and $\mathcal{U}_p$ is undetermined for $p'$ when $p'\neq p$.

Processors are able to \emph{broadcast} messages. To model permissionless protocols, such as Bitcoin, in which each processor has limited ability to broadcast new blocks (and possibly other messages), we require any message broadcast by $p$ to be \emph{permitted} for some public key in $\mathcal{U}_p$: The precise details are as follows. We consider a real-time clock, which exists  outside the system and  measures time in natural number timeslots. The \emph{duration} $\mathcal{D}$ is a determined variable that specifies the set of timeslots (an initial segment of the natural numbers) at which processors carry out instructions. 
At each timeslot $t$, each processor $p$ \emph{receives} a pair $(M,P)$, where either or both of $M$ and $P$ may be empty. Here, $M$ is a finite set of \emph{messages} (i.e.\ strings) that have previously been \emph{broadcast} by other processors. We refer to $M$ as the \emph{message set} received by $p$ at  $t$, and say that each message $m\in M$ is received by $p$ at timeslot $t$. $P$ is referred to as the \emph{permission set} received by $p$ at $t$. Formally, $P$ is a set of pairs, where each pair is of the form $(\mathtt{U},M^{\ast})$ such that $\mathtt{U}\in \mathcal{U}_p$ and $M^{\ast}$ is a potentially infinite set of messages.  If  $(\mathtt{U},M^{\ast})\in P$, then receipt of the permission set $P$  means that each message $m\in M^{\ast}$ may now be permitted for $\mathtt{U}$. This is complicated slightly by our need to model the authenticated setting within an information theoretic model -- we do this by declaring that only $p$ is permitted to  broadcast messages signed by keys in $\mathcal{U}_p$. More precisely, $m\in M^{\ast}$ is permitted for $\mathtt{U}$ if the following conditions are also satisfied: 
\begin{itemize} 
\item $m$ is of the form $(\mathtt{U},\sigma)$ -- thought of as `the message $\sigma$ signed by $\mathtt{U}$'.  
\item For any ordered pair of the form $(\mathtt{U}',\sigma')$ contained in (i.e.\ which is a substring of) $\sigma$, either $\mathtt{U}'\in \mathcal{U}_p$, or else $(\mathtt{U}',\sigma')$ is contained in a message that has been received by $p$.
\end{itemize}
So, as suggested in the above, the latter bulleted conditions allow us to model the fact that we work in the authenticated setting (i.e.\ we assume the use of digital signatures) within an information theoretic computational model.  

To complete the instructions for timeslot  $t$, $p$ then broadcasts a finite set of messages $M'$, each of which must be permitted for some $\mathtt{U}\in \mathcal{U}_p$, makes a  \emph{request set} $R$,  and then enters a new state $x'$, where $x',M'$ and $R$ are determined by the present state $x$ and $(M,P)$, according to the state transition diagram. The form of the request set $R$ will be described shortly, together with how $R$ determines the permission set  received at by $p$ at the next timeslot.

 An \emph{execution} is described by specifying the set of processors, the duration,  the initial states for all processors and by specifying, for each timeslot $t\geq 1$:
\begin{enumerate}
\item The messages and permission sets received by each processor;
\item The instruction that each processor executes, i.e.\ what messages it broadcasts, what requests it makes, and the new state it enters.
\end{enumerate} 

We require that each message is received by $p$ at most once for each time it is broadcast, i.e.\  at the end of the execution it must be possible to specify an injective function $d_p$ mapping each pair $(m,t)$, such that $m$ is received by $p$ at timeslot $t$, to a triple $(p',m,t')$, such that $t'<t$, $p'\neq p$ and such that $p'$ broadcast $m$ at $t'$.

\subsection{The resource pool and the permitter} \label{RP}

 \noindent \textbf{Informal Motivation.} Who should be allowed to create and broadcast new Bitcoin blocks? More broadly, when defining a permissionless protocol,  who should be able to broadcast new messages? For a PoW protocol,  the selection is made depending on computational power. PoS protocols are defined in the context of specifying how to run a currency, and select public keys  according to their stake in the given currency. More generally, one may consider a scarce resource, and then select public keys according to their corresponding resource balance.
 In \cite{lewis2021byzantine}, a framework was introduced according to which protocols run relative to a \emph{resource
pool}, which specifies a resource balance for each public key over the duration of the execution.  
The precise way in which the resource pool is used to determine public key  selection is then black boxed through the use of the \emph{permitter oracle}, to which processors can make requests to broadcast, and which will respond depending on their resource balance. To model Bitcoin, for example, one simply allows each public key to make one request to broadcast a block at each timeslot. The permitter oracle then gives a positive response with probability depending on their resource balance, which in this case is defined by hashrate. So, this gives a straightforward way to model the process, without the need for a detailed discussion of hash functions and how they are used to instantiate the selection process. \\

\noindent \textbf{Formal specification.}  At each timeslot $t$, we refer to the set of all messages that have already been received or broadcast by $p$ as the \emph{message state} of $p$. Each execution happens relative to a (determined or undetermined)
\emph{resource pool},\footnote{As described more precisely in Section \ref{PoWPoS}, whether the resource pool is determined or undetermined will decide whether we are in the \emph{sized} or \emph{unsized} setting.} which in the general case is a function
$\mathcal{R}: \mathcal{U} \times \mathcal{D} \times \mathcal{M}
\rightarrow \mathbb{R}_{\geq 0}$,
where $\mathcal{U}$ is the set of all public keys, $\mathcal{D}$ is the duration and $\mathcal{M}$ is the set of all possible sets of messages. $\mathcal{R}$ can be thought of as specifying the resource
balance of each public key at each timeslot, possibly
relative to a given message state. For each $t$ and $M$, we suppose that certain basic conditions are satisfied:
\begin{enumerate} 
\item[(a)]  If $\mathcal{R}(\mathtt{U},t,M)\neq 0$ then $\mathtt{U}\in \mathcal{U}_p$ for some processor $p$; 
\item[(b)] There are  finitely many $\mathtt{U}$ for which $\mathcal{R}(\mathtt{U},t,M)\neq 0$, and; 
\item[(c)]  $\sum_{\mathtt{U}} \mathcal{R}(\mathtt{U},t,M) >0$. 
\end{enumerate}

 Suppose that, after receiving messages and a permission set at timeslot $t$,  $p$'s message state is $M_0$, and that $M_0^{\ast}$ is the set of all messages that are permitted for $p$ (i.e.\ for some $\mathtt{U}\in \mathcal{U}_p$). We consider two \emph{settings} -- the \emph{timed} and \emph{untimed} settings. The form of each request $r\in R$  made by $p$ at timeslot $t$ depends on the setting, as specified below. While the following definitions might initially seem abstract, shortly we will give examples to make things clear.  
 
 \begin{itemize} 
 
   \item \textbf{The untimed setting}.  Here, each request $r$ made  by  $p$ must be of the form $(\mathtt{U},M,A)$, where
 $\mathtt{U}\in \mathcal{U}_p$, $M \subseteq M_0 \cup M_0^{\ast}$, and where $A$ is
some (possibly empty) extra data. 
The permitter oracle will respond with a pair $(\mathtt{U},M^{\ast})$, where $M^{\ast}$ is a set of strings that may be empty. The value of $M^{\ast}$ 
 will be assumed to be a probabilistic function of the determined variables,   $(\mathtt{U},M,A)$, and of
$\mathcal{R}(\mathtt{U},t,M)$, subject to the condition that $M^{\ast}=\emptyset$ if $\mathcal{R}(\mathtt{U},t,M)=0$. 
If modelling Bitcoin, for example, $M$ might be a set of blocks that have been received by $p$, or that $p$ is already permitted to broadcast, while $A$  specifies a new block extending the `longest chain' in $M$. If the block is valid, then the permitter oracle will give permission to broadcast it with probability depending on the resource balance of $\mathtt{U}$ at time $t$. We will expand on this example below.

 \item \textbf{The timed setting}. Here,  each request $r$ made  by  $p$ must be of the form $(t',\mathtt{U},M,A)$, where $t'$ is a timeslot, 
and  where $\mathtt{U}$, $M$ and $A$ are as in the untimed setting, 
The response  $(\mathtt{U},M^{\ast})$ of the permitter oracle
 will be assumed to be a probabilistic function of the determined variables,  $(t',\mathtt{U},M,A)$, and of
$\mathcal{R}(\mathtt{U},t',M)$, subject to the condition that $M^{\ast}=\emptyset$ if $\mathcal{R}(\mathtt{U},t',M)=0$.  
 \end{itemize} 
The permission  set received by $p$ at timeslot $t+1$ is the set all of responses from the permitter oracle to $p$'s requests at timeslot $t$. 
 

To understand these definitions, it is instructive to consider how they can be used to give a simple model for Bitcoin.  To do so, we work in the untimed setting, and we define the set of possible messages to be the set of possible \emph{blocks}. For each $\mathtt{U}\in \mathcal{U}_p$, we then allow  $p$ to make a single request of the form $(\mathtt{U},M,A)$ at each timeslot. As mentioned above,  $M$ will be a set of blocks that have been received by $p$, or that $p$ is already permitted to broadcast. The entry $A$ will be data (without PoW attached) that specifies a block extending the `longest chain' in $M$. If $A$ specifies a valid block, then the permitter oracle will give permission to broadcast the block specified by $A$ with probability depending on the resource balance of $\mathtt{U}$ at time $t$ (which is determined by hashrate, and is independent of $M$). So, the higher $\mathtt{U}$'s resource balance at a given timeslot, the greater the probability $p$ will be able to mine a block at that timeslot. Of course, a non-faulty processor $p$ will always submit requests of the form $(\mathtt{U},M,A)$, for which $M$ is $p$'s (entire) message state, and such that $A$ specifies a valid block extending the longest chain in $M$.\footnote{So, in this simple model, we don't deal with any notion of a `transaction'. It is clear, though, that the model is sufficient to be able to define what it means for blocks to be \emph{confirmed}, to define notions of \emph{liveness} (roughly, that the set of confirmed blocks grows over time with high probability) and \emph{security} (roughly, that with high probability, the set of confirmed blocks is monotonically increasing over time), and to prove liveness and security for the Bitcoin protocol in this model (by importing existing proofs, such as that in \cite{garay2018bitcoin}). }
 
 The motivation for considering the timed as well as the untimed setting stems from one of the  qualitative differences between PoS and PoW protocols.  PoS protocols are best modelled in the timed setting, where processors can look ahead to determine their permission to broadcast at future timeslots (when their resource balance may be different than it is at present), i.e.\ with PoS protocols, blocks will often have timestamps that cannot be manipulated, and at a given timeslot, a processor may already be able to determine that they have permission to broadcast blocks with a number of different future timestamps.  This means that, when modelling PoS protocols, processors have to be able to make requests corresponding to timeslots $t'$ other than the current timeslot $t$.  We will specify further differences between the timed and untimed settings in Section \ref{PoWPoS}.
 

By a \emph{permissionless protocol} we mean a pair $(\mathtt{S},\mathtt{O})$, where $\mathtt{S}$ is a state transition diagram to be followed by all non-faulty processors, and where $\mathtt{O}$ is a permitter oracle, i.e.\ a probabilistic function of the form described for the timed and untimed settings above. It should be noted that the roles of the resource pool and the
permitter oracle are different, in the following sense: While the resource pool is
a variable (meaning that a given protocol will be expected to function with respect to all possible resource pools consistent with the setting\footnote{Generally, protocols will be considered in a setting that restricts the set of resource pools in certain ways, such as limiting the resource balance of the \emph{adversary}.}), the permitter is part of the protocol description.

\subsection{The adversary and the synchronous and partially synchronous settings} \label{ps} 

While all non-faulty processors follow the state transition diagram $\mathtt{S}$ specified for the protocol, we allow a single undetermined processor $p_A$ to display Byzantine faults, and we think of $p_A$ as being controlled by the \emph{adversary}: In formal terms, the difference between $p_A$ and other processors is that the state transition diagram for $p_A$ might not be $\mathtt{S}$.  Placing bounds on the power of the adversary means limiting their resource balance (since $\mathcal{U}_{p_A}$ is infinite, it does not limit the adversary that they control a single processor). For $q\in [0,1]$, we say the adversary is $q$-\emph{bounded} if their total resource balance is always at most a $q$ fraction of the total, i.e.\ for all $M,t$,  $\sum_{\mathtt{U}\in \mathcal{U}_{p_A}} \mathcal{R}(\mathtt{U},t,M)\leq q\cdot \sum_{\mathtt{U}\in \mathcal{U}} \mathcal{R}(\mathtt{U},t,M) $. 

It is standard in the distributed computing literature \cite{lynch1996distributed} to consider a
variety of \emph{synchronous, partially synchronous}, or
\emph{asynchronous} settings, in which  message delivery might be reliable or subject to various forms
of failure.  We will suppose
that the duration is divided into intervals that are labelled either \emph{synchronous} or
\emph{asynchronous} (meaning that each timeslot is either
synchronous or asynchronous).  We will suppose that during
asynchronous intervals messages can be arbitrarily delayed or not delivered at all. During synchronous intervals, however, we will suppose that messages are always delivered within $\Delta$ many timeslots. So if $t_1\leq t_2$, $m$ is broadcast by $p$ at $t_1$, if $p'\neq p$ and $[t_2,t_2+\Delta]$ is a synchronous interval contained in $\mathcal{D}$, then $p'$ will receive $m$  by timeslot $t_2+\Delta$. Here $\Delta$ is a determined variable. 
 
  We then distinguish two
\emph{synchronicity settings}. In the \emph{synchronous} setting it is
assumed that there are no asynchronous intervals during the duration,
while in the \emph{partially synchronous} setting there may be
\emph{undetermined} asynchronous intervals.

It will be useful to consider the notion of a \emph{timing rule}, by which we mean a partial  function $\mathtt{T}$ mapping tuples of the form $(p,p',m,t)$ to timeslots. We say that an execution follows the timing rule $\mathtt{T}$ if the following holds for all processors $p$ and $p'$:  We have that $p'$ receives $m$ at $t'$  iff there exists some $p$ and $t<t'$ such that $p$ broadcasts the message $m$ at  $t$ and $\mathtt{T}(p,p',m,t)\downarrow =t'$. We restrict attention to timing rules  which are consistent with the setting. Since protocols will be expected to behave well with respect to all timing rules consistent with the setting, it will sometimes be useful to \emph{think of} the adversary as also having control over the choice of timing rule.

\subsection{The structure of the blockchain} \label{blockstructure}

Amongst all broadcast messages, there is a distinguished set referred
to as \emph{blocks}, and one block which is referred to as the
\emph{genesis block}.  Unless it is the genesis block, each block $B$
has a unique \emph{parent} block $\texttt{Par}(B)$, which must be
uniquely specified within the block message. Each block is signed and broadcast by
a single key, $\texttt{Miner}(B)$, but may contain other broadcast
messages which have been signed and broadcast by other keys.  No block can be
broadcast by the processor $p$ that controls $\texttt{Miner}(B)$ at a point strictly prior
to that at which its parent enters $p$'s message state (it is convenient to consider the genesis block a member of all message states at all timeslots). 
$\texttt{Par}(B)$ is defined to be an \emph{ancestor} of $B$, and all
of the ancestors of $\texttt{Par}(B)$ are also defined to be ancestors
of $B$. If $B$ is not the genesis block, then it must have the genesis
block as an ancestor. At any point during the duration, the set of
broadcast blocks thus forms a tree structure.  If $M$ is a set of messages, then we say that it is \emph{downward closed} if it
contains the parents of all blocks in $M$. By a \emph{leaf} of $M$, we
mean a block in $M$ which is not a parent of any block in
$M$. If $M$ is downward closed set of blocks and contains a single leaf, then we
say that $M$ is a \emph{chain}.
 
 \textbf{Generalising the model to DAGs.} It is only for the sake of simplicity that we assume each block has a unique parent block.
The model is chosen to be a sweet spot of being expressible enough to capture many different types of blockchains and not so cumbersome as to obscure the main issues. Only small modifications are then required to deal with DAGS etc.
\subsection{The extended protocol and the meaning of probabilistic
  statements} \label{probs}

To define what it means for a protocol to be secure or live,
we first need a \emph{notion of confirmation} for blocks. This is a
function $ \mathtt{C}$ mapping any message state to a chain that is a
subset of that message state, in a manner that depends on the
 protocol inputs, including a parameter $\varepsilon
>0 $ called the \emph{security parameter}.  The intuition behind
$\varepsilon$ is  that it should upper bound the probability of false
confirmation. Given any message state, $ \mathtt{C}$ returns the set
of confirmed blocks.

In Section \ref{RP}, we stipulated that a permissionless protocol is a  pair $\mathtt{P}=(\mathtt{S},\mathtt{O})$. 
 In general, however, a protocol might only be considered to run relative to a specific notion of confirmation $\mathtt{C}$. We will refer to the triple $(\mathtt{S},\mathtt{O},\mathtt{C})$ as the \emph{extended protocol}.   Often we will suppress explicit mention of $\mathtt{C}$, and assume it to be implicitly attached to a given protocol. We will talk about a protocol being live, for example, when it is really the extended protocol to which the definition applies. It is important to understand, however, that the notion of confirmation $\mathtt{C}$ is separate from $\mathtt{P}$, and does not impact the instructions of the protocol. In principle, one can run the same Bitcoin protocol relative to a range of different notions of confirmation. While the set of confirmed blocks might depend on $\mathtt{C}$, the instructions of the protocol do not, i.e.\ with Bitcoin, one can require five blocks for confirmation or ten, but this does not affect the process of building the blockchain.

For a given permissionless protocol, another way to completely specify an execution (beyond that described in Section \ref{cm}) is via the following breakdown: 
\begin{enumerate} 
\item[(I1)] The determined variables (such as $\Delta$ and $\varepsilon$);
\item[(I2)] The set of processors and their public keys; 
\item[(I3)] The state transition diagram for the adversary $p_A$; 
\item[(I4)] The resource pool (which may or may not be undetermined); 
\item[(I5)] The timing rule;
\item[(I6)] The probabilistic responses of the permitter. 
\end{enumerate} 

With respect to the extended protocol $(\mathtt{S},\mathtt{O},\mathtt{C})$, we call a particular set of choices for (I1)- (I5) a \emph{protocol instance}. Generally, when we discuss an extended protocol, we do so within the context of a \emph{setting}, which constrains the set of possible protocol instances.  The setting might restrict the set of resource pools to those in which the adversary is given a limited resource balance, for example. When we make a probabilistic statement to the effect that a certain condition holds with at most/least a certain probability, this means that the probabilisitic bound holds for all protocol instances consistent with the setting. 
Where convenient, we may also refer to the pair $(\mathtt{P},\mathtt{C})$ as the extended protocol, where $\mathtt{P}=(\mathtt{S},\mathtt{O})$.

\subsection{Defining the timed, sized and multi-permitter settings}  \label{PoWPoS}

In Section \ref{RP}, we gave an example to show how the framework of \cite{lewis2021byzantine}  can be used to model a PoW protocol like Bitcoin. In that context the resource pool is a function $\mathcal{R}: \mathcal{U} \times \mathcal{D}  \rightarrow \mathbb{R}_{\geq 0}$, which is best modelled as undetermined, because one does not know in advance how the hashrate of each public key (or even the total hashrate) will vary over time. The first major difference for a PoS protocol is that the resource balance of each public key now depends on the message state (as is also the case for some proof-of-space protocols, depending on the implementation), and may also be a function of time.\footnote{It is standard practice in PoS blockchain protocols to require a participant to have a currency balance that has been recorded in the blockchain for at least a certain minimum amount of time before they can produce new blocks, for example. So, a given participant may not be permitted to extend a given chain of blocks at timeslot $t$,  but may be permitted to extend the same chain at a later timeslot $t'$.} So the resource pool is a function $\mathcal{R}: \mathcal{U} \times \mathcal{D} \times \mathcal{M}
\rightarrow \mathbb{R}_{\geq 0}$. A second difference is that $\mathcal{R}$ is determined, because one knows from the start how the resource balance of each participant depends on the message state as a function of time. Note that advance knowledge of $\mathcal{R}$ \emph{does not} mean that one knows from the start which processors will have large resource balances throughout the execution, unless one knows which messages will be broadcast. A third difference, to which we have already alluded, is that PoS protocols are best modelled in the timed setting. A fourth difference is that PoW protocols are best modelled by allowing a single request to the oracle for each public key at each timeslot, while this is not necessarily true of PoS protocols. 

In \cite{lewis2021byzantine}, the sized/unsized, timed/untimed,  and single/multi-permitter settings were defined to succinctly capture these differences. The idea is that all permissionless protocols run relative to a resource pool and the difference between PoW and PoS  and other permissionless protocols is whether we are working in the  sized/unsized, timed/untimed,  and single/multi-permitter settings.  If one then comes to consider a new form of protocol, such as proof-of-space,  theorems that have been proved for all protocols in the unsized setting (for example) will still apply,  so long as these new protocols are appropriately modelled in that setting. 
 So the point of this approach is that, by blackboxing the precise mechanics of the processor selection process (whereby processors are selected to do things like broadcast new blocks of transactions), we are able to focus instead on \emph{properties} of the selection process that are relevant for protocol design. This allows for the development of a general theory that succinctly describes the relevant merits of different forms of protocol. 
 The  sized/unsized, timed/untimed,  and single/multi-permitter settings are defined below.

\begin{enumerate} 
\item  \textbf{The timed and untimed settings}.  There are two differences between the timed and untimed settings. The first concerns the form of requests, as detailed in Section \ref{RP}. We also require that the following holds in the timed setting: For each broadcast message $m$, there exists a unique timeslot $t_m$  such that permission to broadcast $m$ was given in response to some request $(t_m,\mathtt{U},M,A)$, and $t_m$  is computable from $m$. We call $t_m$ the \emph{timestamp} of $m$.

\item  \textbf{The sized and unsized settings}.  We call the setting \emph{sized} if the resource balance is determined. By the \emph{total resource balance} we mean the function $\mathcal{T}: \mathbb{N} \times \mathcal{M} \rightarrow \mathbb{R}_{>0}$ defined by $\mathcal{T}(t,M):= \sum_{\mathtt{U}} \mathcal{R}(\mathtt{U},t,M)$.  For the unsized setting,
$\mathcal{R}$ and $\mathcal{T}$ are undetermined, with the only restrictions being:
\begin{enumerate} 
\item[(i)]  $\mathcal{T}$ only takes values in a determined interval 
 $[\alpha_0,\alpha_1]$, where $\alpha_0>0$ (meaning that, although $\alpha_0$ and $\alpha_1$ are determined, protocols will be required to function for all possible $\alpha_0>0$ and $\alpha_1>\alpha_0$, and for all undetermined $\mathcal{R}$ consistent with $\alpha_0,\alpha_1$, subject to (ii) below).\footnote{We consider resource pools with range restricted in this way, because it turns out to be an overly strong condition to require a protocol to function without
 \emph{any} further conditions on the resource pool, beyond
 the fact that it is a function to $\mathbb{R}_{\geq 0}$. Bitcoin will certainly fail  if the total resource balance decreases
sufficiently quickly over time, or if it increases too quickly, causing blocks to be produced too quickly compared to $\Delta$.}
\item[(ii)] There may also be bounds placed on the resource balance of public keys owned by the adversary. 
\end{enumerate} 

\item  \textbf{The multi-permitter and single-permitter settings}. In the \emph{single-permitter} setting, each processor may submit a single request of the form $(\mathtt{U},M,A)$ or $(t,\mathtt{U},M,A)$ (depending on whether we are in the timed setting or not) for each $\mathtt{U}\in \mathcal{U}_p$ at each timeslot, and it is allowed that $A\neq \emptyset$. In the \emph{multi-permitter} setting, processors can submit any number of requests for each key at each timeslot, but they must all satisfy the condition that $A=\emptyset$.
\end{enumerate} 

PoW protocols will generally be best modelled in the untimed, unsized and single-permitter settings. They are best modelled in the untimed setting, because a processor's probability of being granted permission to broadcast a block at timeslot $t$ (even if that block has a different timestamp) depends on their resource balance at $t$, rather than at any other timeslot. They are best modelled in the unsized setting, because one does not know in advance of the protocol execution the amount of mining which will take place at a given timeslot in the future. They are best modelled in the single-permitter setting, so long as permission to broadcast is block-specific. 

PoS protocols are generally best modelled in the timed, sized and multi-permitter settings. They are best modelled in the timed setting, because blocks will generally have non-manipulable timestamps, and because a processor's ability to broadcast a block may be determined at a timestamp $t$ even through the probability of success depends on their resource balance at $t'$ other than $t$.  They are best modelled in the sized setting, because the resource pool is known from the start of the protocol execution. They are best modelled in the multi-permitter setting, so long as permission to broadcast is not block-specific, i.e.\ when permission is granted, it is to broadcast a range of permissible blocks at a given position in the blockchain. 

All of this means that it will generally be straightforward to classify protocols with respect to the theorems from this paper that apply to them. Since Bitcoin and Prism \cite{bagaria2019prism} are PoW protocols, for example, Theorem \ref{t2} applies to those protocols. Since Snow White, Ouroboros \cite{kiayias2017ouroboros} and Algorand are PoS protocols, Theorems \ref{strongsame} and \ref{t3} apply to those protocols. Note that there are a large number of protocols, such as Tendermint \cite{buchman2016tendermint} and Hotstuff \cite{yin2019hotstuff}, which are formally described as permissioned protocols, but which can be implemented as PoS protocols so that  Theorems  \ref{strongsame} and \ref{t3} will then apply.

\subsection{Defining liveness}   \label{defl}
There are a number of papers that successfully describe liveness and security notions for blockchain protocols \cite{garay2018bitcoin,WHGSW16}. Our interest here is in identifying the simplest definitions that suffice to express our later results. To this end, it will be convenient to give a definition of liveness that is more fine-grained than previous definitions, in the sense that it allows us to separate out the security parameter and the number of timeslots in the duration (in previous accounts the number of timeslots in the duration is a function of the security parameter).  
Consider a protocol  with a  notion of confirmation $\mathtt{C}$, and let $|\mathtt{C}(M)|$ denote the number of blocks in $\mathtt{C}(M)$ for any message state $M$. For  timeslots $t_1<t_2$,  let $l_1$ be the maximum value $|\mathtt{C}(M_1)|$ for any $M_1$ which is a message state of any processor at any timeslot $t\leq t_1$, and let $l_2$ be the minimum value $|\mathtt{C}(M_2)|$ for any $M_2$ which is a message state of any processor at timeslot $t_2$.
We say that $[t_1,t_2]$ is a \emph{growth interval}  if $l_2>l_1$. For any duration $\mathcal{D}$, let $|\mathcal{D}|$ be the number of timeslots in $\mathcal{D}$. For $\ell_{\varepsilon,\mathcal{D}}$ which takes values in $\mathbb{N}$ depending on $\varepsilon$ and $\mathcal{D}$, let us say that $\ell_{\varepsilon,\mathcal{D}}$ is \emph{sublinear} in $\mathcal{D}$ if, for each $\varepsilon>0$ and each $\alpha \in (0,1)$,  $\ell_{\varepsilon,\mathcal{D}}<\alpha |\mathcal{D}|$ for all sufficiently large values of $|\mathcal{D}|$ (the motivation for considering sublinearity will be described shortly).

\begin{definition} \label{live} 
A protocol is \textbf{live} if, for every choice of security parameter  $\varepsilon>0$ and duration $\mathcal{D}$, there exists $\ell_{\varepsilon, \mathcal{D}}$, which is sublinear in $\mathcal{D}$, and such that for each pair of timeslots  $t_1< t_2\in \mathcal{D} $ the following holds with probability at least $1-\varepsilon$:  If  $t_2-t_1\geq \ell_{\varepsilon,\mathcal{D}}$ and $[t_1,t_2]$ is entirely synchronous, then $[t_1,t_2]$ is a growth interval.
\end{definition}

\noindent So, roughly speaking, a protocol is live if the number of
confirmed blocks can be relied on to grow during synchronous
intervals of sufficient length. The reason we require
$\ell_{\varepsilon, \mathcal{D}}$ to be sublinear in $\mathcal{D}$ is
so that the number of confirmed blocks likely increases with
sufficient increase in synchronous duration.
For example, a protocol that confirms a block with probability only
$2^{-|\mathcal{D}|}$ at each timeslot  should not be considered live.
  Note also, that while Definition \ref{live} only refers
explicitly to protocols, it is really the \emph{extended protocol} to
which the definition applies. The following stronger notion will also
be useful.

\begin{definition} \label{slive} 
A protocol is \textbf{uniformly live} if, for every choice of security parameter  $\varepsilon>0$ and duration $\mathcal{D}$, there exists $\ell_{\varepsilon, \mathcal{D}}$, which is sublinear in $\mathcal{D}$, and such that the following holds with probability at least $1-\varepsilon$: For all pairs of  timeslots $t_1< t_2\in \mathcal{D}$, if $t_2-t_1\geq \ell_{\varepsilon,\mathcal{D}}$ and $[t_1,t_2]$ is entirely synchronous, then $[t_1,t_2]$ is a growth interval.
\end{definition}

The difference between being live and uniformly live is that the
latter definition requires that, with probability at least
$1-\varepsilon$, \emph{all} appropriate intervals are growth
intervals. The former definition only requires the probabilistic bound
to hold for each interval individually. The reader's immediate
reaction might be that it should follow from the Union Bound that
Definitions \ref{live} and \ref{slive} are essentially
equivalent. This is not so. Firstly, this is because the protocol and
notion of confirmation take the security parameter $\varepsilon$ as
input. Nevertheless, one might think that if a protocol is live then a
`recalibration', which takes some appropriate transformation of the
security parameter as input, should necessarily be uniformly
live. This does not follow (in part) because there is no guarantee
that the resulting $\ell_{\varepsilon, \mathcal{D}}$ will be sublinear
in $\mathcal{D}$ -- see Section \ref{sh} for a detailed analysis.

\subsection{Defining security} \label{sec}

Roughly speaking, \emph{security} requires that confirmed blocks
normally belong to the same chain. Let us say that two distinct blocks
are \emph{incompatible} if neither is an ancestor of the other, and
are {\em compatible} otherwise.  Suppose that, for some processor $p$,  the message state at  $t$ is $M$. If $B\in \mathtt{C}(M)$,  
then we say that $B$ is confirmed for $p$ at $t$.

\begin{definition}[Security] \label{1-secure} A protocol is \textbf{secure} if the following holds for every choice of security parameter  $\varepsilon>0$, for every $p_1,p_2$ and for all timeslots $t_1,t_2$ in the duration:  With probability $> 1-\varepsilon$, all blocks which are confirmed for $p_1$ at $t_1$ are compatible with all those which are  confirmed for  $p_2$ at~$t_2$.   
\end{definition}

%

\noindent The following stronger notion will also be useful. 

\begin{definition}[Uniform Security] \label{3-secure} A protocol is \textbf{uniformly secure} if the following holds for every choice of security parameter  $\varepsilon>0$:  With probability $> 1-\varepsilon$, there do not exist incompatible blocks $B_1,B_2$, timeslots $t_1,t_2$ and $p_1,p_2$ such that $B_i$ is confirmed for $p_i$ at $t_i$ for $i\in \{ 1,2 \}$.
\end{definition}

\noindent The difference between security and uniform security is that the latter requires the probability of even a single disagreement to be bounded, while the former only bounds the probability of disagreement for  each pair of processors at each timeslot pair. Just as for liveness and uniform liveness, it does not follow from the Union Bound that security is essentially equivalent to uniform security. In Section \ref{sh} we will perform a detailed analysis of the relationship between these notions.

\section{Certificates in the partially synchronous setting} \label{pss}

The definitions of this and subsequent sections are all new to this paper, unless explicitly stated otherwise. The rough idea is that `certificates' should be proofs of confirmation. Towards formalising this idea, let us first consider a version which is too weak.

\begin{definition} \label{subj}  If $B\in \mathtt{C}(M)$ then we refer to $M$ as a \textbf{subjective certificate} for $B$. 
\end{definition}

\noindent We will say that a set of messages $M$ is broadcast if every member is broadcast, and that $M$ is broadcast by timeslot $t$ if every member of $M$ is broadcast at a timeslot $\leq t$ (different members potentially being broadcast at different timeslots). If $M$ is a subjective certificate for $B$, then there might exist $M'\supset M$ for which $B\notin \mathtt{C}(M')$. So the fact that  $M$ is broadcast does not constitute proof that $B$ is confirmed with respect to any processor. When do we get harder forms of proof than subjective certificates? Definition \ref{cert} below gives a natural and very simple way of formalising this. 

\begin{definition} \label{cert} 
We say that a protocol with a notion of confirmation $\mathtt{C}$ \textbf{produces certificates} if the following holds with probability $>1-\varepsilon$ when the protocol is run with security parameter $\varepsilon$: There do not exist incompatible blocks $B_1,B_2$, a timeslot $t$ and $M_1,M_2$ which are broadcast by $t$, such that $B_i\in \mathtt{C}(M_i)$ for $i\in \{ 1,2 \}$.  
\end{definition} 
\noindent It is important to stress that, in the definition above, the
$M_i$'s are not necessarily the message states of any processor, but are
rather arbitrary subsets of the set of all broadcast messages. The
basic idea is that, if a protocol produces certificates, then
subjective certificates constitute proof of confirmation. Algorand is
an example of a protocol which produces certificates: The protocol is
designed so that it is unlikely that two incompatible blocks will be
produced at any point in the duration together with appropriate
committee signatures verifying confirmation for each.

Our next aim is to show that, in the partially synchronous setting, producing certificates is equivalent to security.
In fact, producing certificates is clearly at least as strong as uniform security, so it suffices to show that if a protocol is secure then it must produce certificates.  

\begin{theorem} \label{strongsame} If a protocol is secure in the partially synchronous setting then it produces certificates. 
\end{theorem} 
\begin{proof} Towards a contradiction, suppose that the protocol with notion of confirmation $\mathtt{C}$ is secure in the partially synchronous setting, but that there exists a protocol instance\footnote{See Section \ref{probs} for the definition of a protocol instance.} $\mathtt{In}_1$  with security parameter $\varepsilon$, such that the following holds with probability $\geq \varepsilon$:  There exist incompatible blocks $B_1,B_2$, a timeslot $t$ and $M_1,M_2$ which are broadcast by $t$, such that $B_i\in \mathtt{C}(M_i)$ for $i\in \{ 1,2 \}$. This means that the following holds with probability $\geq \varepsilon$ for $t_{\text{last}}$, which is the last timeslot in the duration: There exist incompatible blocks $B_1,B_2$ and $M_1,M_2$ which are broadcast by $t_{\text{last}}$, such that $B_i\in \mathtt{C}(M_i)$ for $i\in \{ 1,2 \}$.
Consider the protocol instance $\mathtt{In}_2$ which has the same values for determined variables as $\mathtt{In}_1$,  the same state transition diagram for the processor of the adversary and the same set of processors with the same set of  public keys, except that now there are two extra processors $p_1$ and $p_2$. Suppose that the resource pool for $\mathtt{In}_2$ is the same as that for $\mathtt{In}_1$ when restricted to public keys other than those in $\mathcal{U}_{p_1}$ and $\mathcal{U}_{p_2}$, and that all keys in $\mathcal{U}_{p_1}$ and $\mathcal{U}_{p_2}$ have zero resource balance throughout the duration. Suppose further, that the timing rule for $\mathtt{In}_2$ is the same as that for $\mathtt{In}_1$ when restricted to tuples $(p,p',m,t)$  such that $p\notin \{  p_1,p_2 \}$ and  $p'\notin \{ p_1,p_2 \}$, but that now all timeslots are asynchronous.  
 According to the definition of Section \ref{RP}, and since all keys in $\mathcal{U}_{p_1}$ and $\mathcal{U}_{p_2}$ have zero resource balance throughout the duration,  it follows by induction on timeslots that the probability distribution on the set of broadcast messages is the same at each timeslot for $\mathtt{In}_2$ as for $\mathtt{In}_1$, independent of which messages are received by $p_1$ and $p_2$.  
It therefore holds  for the protocol instance $\mathtt{In}_2$ that with probability $\geq \varepsilon$ there  exist incompatible blocks $B_1,B_2$,  and $M_1,M_2$ which are broadcast by $t_{\text{last}}$, such that $B_i\in \mathtt{C}(M_i)$ for $i\in \{ 1,2 \}$.  Now suppose that $p_1$ and $p_2$ do not receive any messages until $t_{\text{last}}$, and then receive the message sets  $M_1$ and $M_2$ (if they exist) respectively. This suffices to demonstrate that the definition of security is violated with respect to $t_{\text{last}}$, $\varepsilon$, $p_1$ and $p_2$. 
\end{proof} 

\begin{corollary} \label{collapse}
Security and uniform security are equivalent in the partially synchronous setting.
\end{corollary} 
\begin{proof} 
This follows from Theorem \ref{strongsame} and the fact that  producing certificates clearly implies uniform security.
\end{proof}

\section{Security and uniform security in the synchronous setting} \label{sh}

Having dealt with the partially synchronous setting, our next task is to consider the synchronous setting. To do so, however, we first need to formalise the notion of a \emph{recalibration}.

\subsection{Defining recalibrations} \label{recal} Theorem \ref{strongsame} seems to tie things up rather neatly for the partially synchronous setting. In particular, the equivalence of security and uniform security meant that we were spared having to carry out a separate analysis for each security notion. It is not difficult to see, however, that the two security notions will not be equivalent in the synchronous setting. To see this, we can consider the example of Bitcoin. Suppose first that we operate in the standard way for Bitcoin, and use a notion of confirmation $\mathtt{C}$ that depends only on the security parameter $\varepsilon$, and not on the duration $\mathcal{D}$, so that the number of blocks required for confirmation is just a function of $\varepsilon$. In this case, the protocol is secure in the synchronous setting \cite{garay2018bitcoin}. 
It is also clear, however,  that this protocol will not be uniformly secure in a setting where the adversary controls a non-zero amount of mining power: If a fixed number of blocks are required for confirmation then, given enough time, the adversary will eventually complete a double spend (i.e.\ the adversary will double spend with probability tending to 1 as the number of timeslots tends to infinity). That said, it is also not difficult to see how one might  `recalibrate' the protocol to deal with different durations -- to make the protocol uniformly secure, the number of blocks required for confirmation should be a function of both $\varepsilon$ and $\mathcal{D}$. 

The point of this subsection is to formalise the idea of recalibration and to show that, if a protocol is secure, then (under fairly weak conditions) a recalibration will be uniformly secure. The basic idea is very simple -- one runs the initial (unrecalibrated) protocol for smaller values of $\varepsilon$ as the duration increases, but one has to be careful that the resulting $\ell_{\varepsilon,\mathcal{D}}$ is sublinear in $\mathcal{D}$. 
 
\begin{definition} 
We say  $(\mathtt{P}_2,\mathtt{C}_2)$ is a recalibration of the extended protocol $(\mathtt{P}_1,\mathtt{C}_1)$ if running $\mathtt{P}_2$ given certain inputs means running $\mathtt{P}_1$ for a computable transformation of those inputs, and then terminating after $|\mathcal{D}|$ many steps are complete. 

 \end{definition}

\noindent So, if running $\mathtt{P}_2$ with security parameter
$\varepsilon$ and for $n$ many timeslots means running $\mathtt{P}_1$
with input parameters that specify a security parameter
$\varepsilon/10$ and that specify a duration consisting of  $2n$ many
timeslots, and then terminating after $n$ many timeslots have been
completed, then $\mathtt{P}_2$ is a recalibration of $\mathtt{P}_1$.\footnote{The choices $\varepsilon/10$ and $2n$ are arbitrarily chosen for the purpose of example. The reader might wonder why one should specify a duration of $2n$ timeslots and then terminate after $n$ many. This is because the instructions of the first $n$ timesteps can depend on the intended duration. In Algorand, committee sizes will depend on the intended duration, for example.}  Note also, that we allow the recalibration to use a different notion of confirmation. 

In the following, we say that $\ell_{\varepsilon,\mathcal{D}}$ is independent of $\mathcal{D}$ if 
$\ell_{\varepsilon,\mathcal{D}} = \ell_{\varepsilon,\mathcal{D}'}$ for all $\varepsilon>0$ and all $\mathcal{D},\mathcal{D}'$.  When $\ell_{\varepsilon,\mathcal{D}}$  is independent of $\mathcal{D}$, we will often write $\ell_{\varepsilon}$ for $\ell_{\varepsilon,\mathcal{D}}$.
 
 \begin{definition} 
 In the \textbf{bounded user} setting we assume that there is a finite upper bound on the number of processors, which holds for all protocol instances.\footnote{Note that the requirement here is that the number of processors is bounded, rather than the number of public keys.} 
 \end{definition}

\begin{proposition} \label{recal}
Consider the synchronous and bounded user setting. Suppose $\mathtt{P}$ satisfies liveness with respect to $\ell_{\varepsilon,\mathcal{D}}$,  that $\ell_{\varepsilon,\mathcal{D}}$  is independent of $\mathcal{D}$,  and that for each $\alpha >0$, $ \ell_{\varepsilon} <\alpha \varepsilon^{-1}$  for all sufficiently small $\varepsilon>0$.   If $\mathtt{P}$ is secure, there exists a recalibration of $\mathtt{P}$ that is uniformly live and uniformly secure. 
\end{proposition}

The conditions on $\ell_{\varepsilon,\mathcal{D}}$ in the statement of Proposition \ref{recal} can reasonably be regarded as weak, because existing protocols which are not already uniformly secure will normally  satisfy the conditions that: $(\dagger_a)$  $\ell_{\varepsilon,\mathcal{D}}$  is independent of $\mathcal{D}$, and;
$(\dagger_b)$ For some  constant  $c$ and any  $\varepsilon \in (0,1)$,  we have $\ell_{\varepsilon} < c \text{ln} \ \frac{1}{\varepsilon}$. 
The example of Bitcoin might be useful for the purposes of illustration here. Bitcoin is secure in the synchronous setting, and the number of blocks required for confirmation is normally considered to be independent of the duration. The number of blocks required for confirmation \emph{does} depend on how sure one needs to be that an adversary cannot double spend in any given time interval, but it's also true that an adversary's chance of double spending in a given time interval decreases exponentially in the number of blocks required for confirmation as well.  So Bitcoin is an example of a protocol satisfying $(\dagger_a)$ and $(\dagger_b)$ above.

\begin{proof}[Proof of Proposition \ref{recal}] 
It is useful to consider a security notion that is intermediate between security and uniform security. 
For the purposes of the following definition, we say that a block is confirmed at timeslot $t$ if there exists at least one processor for whom that is the case.  

\begin{definition}[Timeslot Security] \label{2-secure} A protocol is \textbf{timeslot secure} if the following holds for every choice of security parameter  $\varepsilon>0$, and for all timeslots $t_1,t_2$ in the duration:  With probability $> 1-\varepsilon$, all blocks which are confirmed at  $t_1$ are compatible with all blocks which are confirmed at $t_2$.   
\end{definition}

\noindent So the difference between timeslot security and uniform security is that the latter requires the probability of even a single disagreement to be bounded, while the former only bounds the probability of disagreement for each  pair of timeslots. Similarly, the difference between security and timeslot security is that, for each pair of timeslots, the latter requires the probability of even a single disagreement to be bounded, while the former only bounds the probability of disagreement for each pair of processors at that timeslot pair.

Now suppose $\mathtt{P}$ is live and secure, and that the conditions of Proposition \ref{recal} hold. Then it follows directly from the Union Bound that, if the number of users is bounded,  then some recalibration of $\mathtt{P}$ is live and  timeslot secure and satisfies the conditions of Proposition \ref{recal}. Since a recalibration of a recalibration of $\mathtt{P}$ is a recalibration of $\mathtt{P}$, our main task is therefore to show that, if $\mathtt{P}$ is live and timeslot secure and the conditions of Proposition \ref{recal} hold, then there exists a recalibration  of $\mathtt{P}$ that is uniformly live and uniformly secure.  

So suppose  $(\mathtt{P},\mathtt{C})$ is live and timeslot secure, and that the conditions of Proposition \ref{recal} hold. 
 Suppose we are given $\varepsilon_0$ and $\mathcal{D}_0$ as inputs to our recalibration $(\mathtt{P}',\mathtt{C}')$. We wish to find an appropriate security parameter $\varepsilon_1$ and a duration $\mathcal{D}_1\geq \mathcal{D}_0$ to give as inputs to $\mathtt{P}$ and $\mathtt{C}$, so that uniform security is satisfied with respect to $\varepsilon_0$ and $\mathcal{D}_0$ if we run $\mathtt{P}$ with  inputs $\varepsilon_1$ and $\mathcal{D}_1$ and then terminate after $|\mathcal{D}_0|$ many timeslots. The difficulty is to ensure that $\ell_{\varepsilon_1}$ remains sublinear in $\mathcal{D}_0$. To achieve this, let $n:=|\mathcal{D}_0|$,  set $\varepsilon_1:=\varepsilon_0/2n$ and choose $|\mathcal{D}_1|> n+ \ell_{\varepsilon_1}$, so that $\mathcal{D}_0$ is the first $n$ timeslots in  $\mathcal{D}_1$. 
 This defines the recalibration. It remains to establish uniform liveness and uniform security.

 For uniform liveness we must have that,  for each $\alpha \in (0,1)$,   $\ell_{\varepsilon_1} <\alpha n$ for all sufficiently large values of $n$  -- if this condition holds then it follows from the Union Bound that our recalibration will satisfy uniform liveness (and the required sublinearity in $\mathcal{D}_0$) with respect to $\ell'_{\varepsilon_0,\mathcal{D}_0}:=\ell_{\varepsilon_1}$.   
 The condition holds since we are given that for each $\alpha >0$, $ \ell_{\varepsilon} <\alpha \varepsilon^{-1}$  for all sufficiently small $\varepsilon>0$. Suppose given $\alpha >0$, and put $\alpha':=\alpha\varepsilon_0/2$. Then we have that, for all sufficiently large $n$:

 \[ \ell_{\varepsilon_1}< \alpha' (\varepsilon_0/2n)^{-1}= \alpha n. \] 

Next we must show that the conditions for uniform security are satisfied.  Suppose $\mathtt{P}$ is given inputs $\varepsilon_1$ and $\mathcal{D}_1$ and is actually run for $|\mathcal{D}_1|$ many timeslots.
We aim to show that, with probability $> 1-\varepsilon_0$, there do not exist incompatible blocks $B_1,B_2$, timeslots $t_1,t_2 \in \mathcal{D}_0$ and $p_1,p_2$ such that $B_i$ is confirmed for $p_i$ at $t_i$ for $i\in \{ 1,2 \}$. 
  Let  $t_{\text{last}}$ be the last timeslot of the duration $\mathcal{D}_1$  and define $t^{\ast}:= t_{\text{last}}-\ell_{\varepsilon_1}$. The basic idea is that the two following conditions hold with high probability: (a)  $[t^{\ast}, t_{\text{last}}]$  is a growth interval, and  (b) There does not exist  $t_1\in \mathcal{D}_0$,  processors $p_1,p_2$ and incompatible blocks $B_1,B_2$, such that  $B_1$ is confirmed for $p_1$ at $t_1$ and $B_2$ is confirmed for $p_2$ at $t_{\text{last}}$.  When both these conditions hold, and since $t^{\ast}>n$, this suffices to show that no incompatible and confirmed blocks exist during the duration $\mathcal{D}_0$. Now let us see that in more detail. 
  
  By the choice of $\mathcal{D}_1$, $t^{\ast}>n$. It follows from the definition of liveness that  $(\dagger_1)$ below fails to hold with probability   $\leq \varepsilon_1$: 
 
\begin{enumerate} 
\item[$(\dagger_1)$] $[t^{\ast}, t_{\text{last}}]$ is a growth interval. 
\end{enumerate}
Note that, so long as $(\dagger_1)$ holds, every user has more confirmed blocks at  $t_{\text{last}}$ than any user does at any timeslot in $\mathcal{D}_0$. It also follows from the Union Bound, and the definition of liveness and timeslot security, that  $(\dagger_2)$ below fails to hold with probability $\leq n \varepsilon_1 =\varepsilon_0/2$: 
\begin{enumerate} 
\item[$(\dagger_2)$] There does not exist  $t_1\in \mathcal{D}_0$,  processors $p_1,p_2$ and incompatible blocks $B_1,B_2$, such that  $B_1$ is confirmed for $p_1$ at $t_1$ and $B_2$ is confirmed for $p_2$ at $t_{\text{last}}$. 
\end{enumerate}

\noindent Now note that: 
\begin{enumerate} 
\item[(a)] If $(\dagger_1)$ and $(\dagger_2)$ both hold, then there do not exist incompatible blocks $B_1,B_2$, timeslots $t_1,t_2\in \mathcal{D}_0$ and $p_1,p_2$ such that $B_i$ is confirmed for $p_i$ at $t_i$ for $i\in \{ 1,2 \}$.
\item[(b)] With probability $>1-\varepsilon_1-\varepsilon_0/2 \geq 1-\varepsilon_0$,  $(\dagger_1)$ and $(\dagger_2)$ both hold. 
\end{enumerate} 
So uniform security is satisfied with respect to $\varepsilon_0$ and $\mathcal{D}_0$, as required. 
 \end{proof}

 \begin{definition} \label{sf} 
 We say $\mathtt{P}$ has \textbf{standard functionality} if it is uniformly live and uniformly secure. We say that a recalibration of $\mathtt{P}$  is \textbf{faithful} if it has standard functionality when $\mathtt{P}$ does. 
 \end{definition} 
 
 Proposition \ref{recal} justifies concentrating on protocols which have standard functionality where it is convenient to do so, since protocols which are live and secure will have recalibrations  with standard functionality, so long as the rather weak conditions of  Proposition \ref{recal} are satisfied. Again, when we talk about the security and liveness of a protocol, it is really the extended protocol that we are referring to. 
  
\section{Certificates in the synchronous setting} \label{cs} 

\subsection{The synchronous and unsized setting} 
As outlined in the introduction, part of the aim of this paper is to give a positive answer to Q3, by showing that whether a protocol produces certificates comes down essentially to properties of the processor selection process. In the unsized setting protocols cannot produce certificates. In the sized setting, recalibrated protocols will automatically produce certificates, at least if they are of `standard form'. For the partially synchronous setting, the results of \cite{lewis2021byzantine} and Section \ref{pss} already bear this out: The sized setting is required for security and all secure protocols must produce certificates. 
The following theorem now deals with the unsized and synchronous setting. Recall that, in the unsized setting,  the total resource balance belongs to a determined 
interval $[\alpha_0,\alpha_1]$. We say that the protocol operates `in the presence of a non-trivial adversary' if the setting allows that the adversary may have resource balance at least $\alpha_0$ throughout the duration.

 \begin{theorem} \label{t2}
 Consider the  synchronous and unsized setting. If a protocol is live then, in the presence of a non-trivial adversary,  it does not produce  certificates. 
 \end{theorem} 
 \begin{proof} 
 The basic idea is that the adversary with resource balance at least $\alpha_0$ can `simulate' their own execution of the protocol, in which only they have non-zero resource balance, while the non-faulty processors carry out an execution in which the adversary does not participate. Simulating their own execution means that the adversary carries out the protocol as usual, while ignoring messages broadcast by the non-faulty processors,  but does not initially broadcast messages when given permission to do so. Liveness (together with the fact that the resource pool is undetermined) guarantees that, with high probability, both the actual and simulated executions produce blocks which look confirmed from their own perspective. These blocks will be incompatible with each other and, once the adversary finally broadcasts the messages that they have been given permission for,  these blocks will all have subjective certificates which are subsets of the set of broadcast messages. This suffices to show that the protocol does not produce certificates.  
 
 More precisely, we consider two instances of the protocol $\mathtt{In}_0$ and $\mathtt{In}_1$ in the synchronous and unsized setting, which have the same values for all determined variables -- including the same sufficiently small security parameter $\varepsilon$ and  the same sufficiently long duration $\mathcal{D}$ -- and also have the same set of processors and the same message delivery rule, but which differ as follows: 
 \begin{itemize} 
 \item In $\mathtt{In}_0$, a set of processors $\mathcal{P}_0$ control public keys in a set $\mathcal{U}_0$, which are the only public keys that do not have zero resource balance throughout the duration. The total resource balance $\mathcal{T}$ has a fixed value, $\alpha $ say. 
 \item In $\mathtt{In}_1$, it is the adversary who controls the public keys in $\mathcal{U}_0$, and those keys have the same resource balance throughout the duration as they do in $\mathtt{In}_0$. Now, however,  another set of processors $\mathcal{P}_1$ control public keys in a set $\mathcal{U}_1$ (disjoint from $\mathcal{U}_0$), and the public keys in $\mathcal{U}_1$ also have total resource balance $\alpha $ throughout the duration, i.e.\ the resource balances of these keys always add to $\alpha$. 
 \end{itemize} 
 In $\mathtt{In}_1$, we suppose that the adversary simulates the processors in $\mathcal{P}_0$ for $\mathtt{In}_0$ (which can be done with the single processor $p_A$), which means that the adversary carries out the instructions for those processors, with the two following exceptions. Until a certain timeslot $t^{\ast}$, to be detailed subsequently, they: 
 \begin{enumerate} 
 \item[(a)]  Ignore all messages broadcast by non-faulty processors, and; 
 \item[(b)] Do not actually broadcast messages when permitted, but consider them received by simulated processors in $\mathcal{P}_0$ as per the message delivery rule. 
 \end{enumerate} 
 
 For $\mathtt{In}_0$ (so long as the duration is sufficiently long), liveness guarantees the existence of a timeslot $t_0$ for which the following holds with probability $>1-\varepsilon$: 
 \begin{enumerate} 
 \item[$(\diamond_0)$] At $t_0$ there exists a set of broadcast messages $M_0$ and a block $B_0$ such that $B_0\in \mathtt{C}(M_0)$. 
 \end{enumerate} 
 
  For $\mathtt{In}_1$, liveness guarantees the existence of a timeslot $t_1$ for which the following holds with probability $>1-\varepsilon$: 
 \begin{enumerate} 
 \item[$(\diamond_1)$] At $t_1$ there exists a set of broadcast messages $M_1$ and a block $B_1$ such that $B_1\in \mathtt{C}(M_1)$. 
 \end{enumerate} 
 
 Choose $t^{\ast}>t_0,t_1$. Our framework stipulates that the instructions of the protocol
for a given user at a given timeslot are a deterministic function of their present state and the message set and permission set received at that timeslot. 
It also stipulates that the response of the permitter to a request
$(t',\mathtt{U},M,A)$ is a probabilistic function of the determined variables,  $(t',\mathtt{U},M,A)$, and of
$\mathcal{R}(\mathtt{U},t',M)$.
Since we are working in the unsized setting, $\mathtt{In}_1$ and $\mathtt{In}_0$ have the same determined variables. It therefore follows
by induction on timeslots $t\leq t^{\ast}$, that the following is true at all points until the end of timeslot $t$:
\begin{enumerate} 
\item[$(\diamond_2)$] The probability distribution for $\mathtt{In}_0$ on the set of permission sets given by the permitter   is identical to the probability distribution for $\mathtt{In}_1$ on the set of permission sets given by the permitter to the adversary. 
\end{enumerate} 

Now suppose that at timeslot $t^{\ast}$ the adversary broadcasts all messages for which they have been given permission by the permitter. Note that, according to the assumptions of Section
\ref{blockstructure}, any block $B_0$ broadcast by the adversary at $t^{\ast}$ will be incompatible with any block $B_1$ that has been broadcast by any honest user up to that point. Combining $(\diamond_0)$, $(\diamond_1)$ and $(\diamond_2)$, we see that (so long as $\varepsilon$ is sufficiently small that $\varepsilon<1-2\varepsilon$) the following holds with probability 
 $>\varepsilon$ for $t^{\ast}$ and $\mathtt{In}_1$: There exist incompatible blocks $B_0,B_1$,  and $M_0,M_1$ which are broadcast by the end of $t^{\ast}$, such that $B_i\in \mathtt{C}(M_i)$ for $i\in \{ 0,1 \}$.
 This suffices to show that the protocol does not produce certificates. 
 \end{proof} 
 
 \subsection{The synchronous and sized setting}\label{ss:recalibration}
 
\noindent \textbf{The example of sized Bitcoin.}  Our aim in this subsection is to show that, if we work in the synchronous and sized setting, and if a protocol is of `standard form', then a recalibration will produce certificates. To make this precise, however, it will be necessary to recognise the potentially \emph{time dependent} nature of proofs of confirmation. To explain this idea,  it is instructive to consider the example of Bitcoin in the sized setting: The protocol is Bitcoin, but now we are told in advance precisely how the hash rate capability of the network varies over time, as well as bounds on the hash rate of the adversary.\footnote{Normally we think of PoW protocols as operating in the unsized setting, precisely because such guarantees on the hash rate are not realistic.} To make things concrete, let us suppose that the total hash rate is fixed over time, and that the adversary has 10\% of the hash rate at all times. Suppose that, during the first couple of hours of running the protocol, the difficulty setting is such that the network as a whole (with the adversary acting honestly) will produce an expected one block every 10 minutes. Suppose further that, after a couple of hours, we  see a block $B$ which belongs to a chain $C$, in which it is followed by 10 blocks. In this case, the constraints we have been given mean that it is very unlikely that $B$ does not belong to the longest chain. So, \emph{at that timeslot}, $C$ might be considered a proof of confirmation for $B$, i.e.\ the existence of the chain $C$ can be taken as proof that $B$ is confirmed. The nature of this proof is time dependent, however. The same set of blocks (i.e. $C$)  a large number of timeslots later would not constitute proof of confirmation. 
 
 If we now consider a PoS version of the example above, modified to work for Snow White rather than Bitcoin, then the proof produced will \emph{not} be time dependent. This is because PoS protocols function in the timed setting, i.e.\ when permission is given to broadcast $m$ in response to a request  $(t,\mathtt{U},M,A)$, other users are able to determine $t$ from $m$.  In order to prove that (recalibrated) protocols in the sized setting produce certificates, we will have to make the assumption that we are also working in the timed setting. \\
  
  \noindent \textbf{Protocols in standard form.} The basic intuition behind the production of certificates in the sized setting can be seen from the example of ``Sized Bitcoin'' above. Once a block is confirmed, non-faulty processors will work `above' this block. So long as those processors possess a majority of the total resource balance, and so long as the permitter reflects this fact in the permissions it gives, then those  non-faulty processors will broadcast a set of messages which suffices (by its existence rather than the fact that it is the full message state of any user) to give proof of confirmation. 
 This proof of confirmation might be temporary, but it will not be temporary in the timed setting.  
 
 This intuitive argument, however, assumes that the protocol satisfies certain standard properties. As alluded to above, there is an assumption that the set of messages broadcast by a group of processors will reflect their resource balances and that the adversary will have a minority resource balance. There is also an assumption that broadcast messages  will (in some sense) point to a particular position on the blockchain. So we will have to formalise these ideas, and the results we prove will only hold modulo the assumption that these standard properties are satisfied.

 First, let us formalise the idea that messages always point to a position on the blockchain. 
 
 \begin{definition} 
 We say that a protocol is in \textbf{standard form} if it satisfies all of the following: 
 \begin{itemize} 
 \item The protocol has standard functionality (see Definition \ref{sf}). 
 \item Every broadcast message is `attached' to a specific block (blocks being attached to themselves). 
 \item While $B$ is confirmed for $p$,  the state transition diagram $\mathtt{S}$ will only instruct $p$ to broadcast messages which are attached to $B$ or descendants of $B$.  
 
 \end{itemize}
 \end{definition}
 
\noindent \textbf{Reflecting the resource pool.}  Now let us try to describe how the permitter might  reflect the resource pool. We will need a simple way to say that one set of processors consistently has a higher resource balance than another. 
 
 \begin{definition} 
 For $\Theta>1$, we say a set of public keys $\mathcal{U}_1$ \textbf{dominates} another set $\mathcal{U}_2$, denoted $\mathcal{U}_1 >_{\Theta} \mathcal{U}_2$, if the following holds for all sets of broadcast messages $M$ and all timeslots $t$: 
 
 \[ \sum_{\mathtt{U}\in \mathcal{U}_1} \mathcal{R}(\mathtt{U},t,M) >\Theta \cdot \sum_{\mathtt{U}\in \mathcal{U}_2} \mathcal{R}(\mathtt{U},t,M). \]
 \end{definition}
 
 Next, we will need to formalise the idea that, if one set of keys dominates another, then they will be able to broadcast discernibly different sets of messages. Recall that, in the timed setting, each message $m$ corresponds to a timeslot $t_m$, which can be determined from $m$. We write $\mathcal{M}[t_1,t_2]$ to denote the set  $\{ M| \ \forall m\in M, \  t_m\in [t_1,t_2] \}$.  We will say that the set of keys $\mathcal{U}_0$ is \emph{directed to broadcast $M$} if, for every $m\in M$,  there is some member of $\mathcal{U}_0$ that is given permission to broadcast $m$  and is directed to broadcast $m$ by the protocol. We will say that $\mathcal{U}_0$ is \emph{able to broadcast $M$} if, for every $m\in M$,  there is some member of $\mathcal{U}_0$ that is given permission to broadcast $m$. We define $\mathcal{M}^{\ast}:= \{ M|\ M \text{ is finite} \}$. We let $\mathbb{T}$ be the set of functions $T: \mathcal{D} \times \mathcal{M} \rightarrow \mathbb{R}_{\geq 0}$ (so that the total resource balance $\mathcal{T}\in \mathbb{T}$).  We say that a set of keys $\mathcal{U}_0$ has total resource balance $T:  \mathcal{D} \times \mathcal{M} \rightarrow \mathbb{R}_{\ge 0}$ if  $T(t,M) = \sum_{\mathtt{U}\in \mathcal{U}_0} \mathcal{R}(\mathtt{U},t,M)$.  In the definition below, we say $r$ is sublinear in $|\mathcal{D}|$ if, \ for each $\Theta, \varepsilon,T$,  and for every $\alpha\in (0,1)$,  it holds that $r(\Theta,\varepsilon,T,|\mathcal{D}|)<\alpha |\mathcal{D}|$ for all sufficiently large $|\mathcal{D}|$.

  \begin{definition} \label{reflect}
 We  say that  $(\mathtt{S},\mathtt{O}, \mathtt{C})$ \textbf{reflects the resource pool} if there exist computable finite valued functions $r: \mathbb{R}_{>1}\times \mathbb{R}_{>0} \times \mathbb{T} \times \mathbb{N} \rightarrow \mathbb{N}$ and $\mathtt{X}: \mathcal{M}^{\ast} \times \mathbb{R}_{>1}\times \mathbb{R}_{>0} \times \mathbb{T} \times \mathbb{N} \rightarrow 2^{\mathcal{M}^{\ast}}$, such that: 
 \begin{enumerate}  
\item $r$ is sublinear in $|\mathcal{D}|$. 
 
 \item If $\mathcal{U}_1 \cup \mathcal{U}_2$ has total resource balance $T$, and if $\mathcal{U}_1>_{\Theta} \mathcal{U}_2$, then, when the protocol is run with security parameter $\varepsilon$ and for $|\mathcal{D}|$ many timeslots, the following holds with probability $>1-\varepsilon$: For all intervals of timeslots $[t_1,t_2]$ with $t_2-t_1\geq r(\Theta,\varepsilon,T,|\mathcal{D}|)$, there exists some $M\in \mathcal{M}[0,t_1)$ and an element of   $\mathtt{X}(M,\Theta,\varepsilon,T, |\mathcal{D}|) \cap \mathcal{M}[t_1,t_2]$ which $\mathcal{U}_1$ is directed to broadcast, while there does not exist any $M'\in \mathcal{M}[0,t_1)$ which is broadcast and such that $\mathcal{U}_2$ is able to broadcast some element of $\mathtt{X}(M',\Theta,\varepsilon,T,|\mathcal{D}|) \cap \mathcal{M}[t_1,t_2]$. 
 \end{enumerate} 
 \end{definition}
 
 So in Definition \ref{reflect}, $r$ specifies a number of timeslots. Then $\mathtt{X}$ specifies certain sets of messages $M$ such that, if $\mathcal{U}_1>_{\Theta} \mathcal{U}_2$ and $\mathcal{U}_1 \cup \mathcal{U}_2$ has total resource balance $T$, then $\mathcal{U}_1$ can be expected to broadcast one of these sets $M$ in any interval of sufficient length (i.e.\ the length specified by $r$). To make this interesting, we also have that $\mathcal{U}_2$ can be expected \emph{not} to make such broadcasts. 
 To see why this is a natural and reasonable condition to assume, it is instructive to consider the example of Sized Bitcoin. Suppose that in some execution the honest users always have at least 60\% of the mining power. Then, over any long period of time $r$, we can be fairly sure that honest users will get to make at least 50\% of the expected number of block broadcasts, while the adversary is unlikely to be able to make such broadcasts if $r$ is large enough. In fact, the exponentially fast convergence for the law of large numbers guaranteed by bounds like Hoeffding's inequality, means $r$ only needs to grow with $\text{ln} \ 1/p$, where $p$ is the probability of error (i.e.\ the probability these conditions on the block broadcasts don't hold in a given interval).  It is therefore not difficult to see that Sized Bitcoin would reflect the resource pool if it could be implemented in a timed setting. Similar arguments can be made for all well known PoS protocols,\footnote{The example of Snow White was discussed previously. As suggested in Section \ref{intro}, one way to define $\mathtt{X}$ in the context of Snow White is to consider 
long chains of sufficient \emph{density}, meaning that they have members
corresponding to most possible timeslots, that they cannot likely be
produced by a (sufficiently bounded) adversary.} and these \emph{are} implemented in the timed setting.

 \begin{definition} 
 In the \textbf{bounded adversary} setting it is assumed that: 
 \begin{enumerate} 
 \item[(i)]   $\mathcal{U}_1 >_{\Theta} \mathcal{U}_2$ for some determined input parameter $\Theta>1$, where  $\mathcal{U}_1$ is the set of keys controlled by non-faulty processors, and $\mathcal{U}_2$ is the the set of keys controlled by the adversary. 
 \item[(ii)]  $(\mathtt{S},\mathtt{O},\mathtt{C})$  reflects the resource pool. 
 \end{enumerate} 
 \end{definition} 

 Finally, we can now formalise the idea that under standard conditions,  standard protocols in the sized setting produce certificates. 

\begin{theorem} \label{t3}
Consider  the timed,  bounded adversary and  sized setting.  If $\mathtt{P}$  is in standard form, then there exists a faithful recalibration that produces certificates. 
\end{theorem}
\begin{proof} 
To define our recalibration $(\mathtt{P}',\mathtt{C}')$, suppose we are given values for $\varepsilon,\mathcal{T},\Theta$ and $\mathcal{D}$. We need to specify a value $\varepsilon'$ to give as input to $\mathtt{P}$ (we will leave other values unchanged), and we must also define $\mathtt{C}'$. Then we need to show that the new extended protocol  is uniformly live and produces certificates. 

We define $\varepsilon':=\varepsilon/4$. Towards defining $\mathtt{C}'$, suppose that $\mathtt{P}$ satisfies uniform liveness with respect to $\ell_{\varepsilon',\mathcal{D}}$. We divide the duration into intervals of length  $>\ell_{\varepsilon',\mathcal{D}}$, by defining $t_i:= i \cdot (\ell_{ \varepsilon',\mathcal{D}}+ r(\Theta,\varepsilon',\mathcal{T},|\mathcal{D}|))$.
 From the definition of uniform liveness we have the following. 
\begin{enumerate} 
\item[$(\$_1)$] With probability $>1-\varepsilon/4$ it holds that,  for all $i$  with $t_i\leq |\mathcal{D}|$, all users have at least $i$ many confirmed blocks  by the end of  timeslot $t_i$. 
\end{enumerate} 

Now suppose  $(\mathtt{P}, \mathtt{C})$ satisfies Definition \ref{reflect} with respect to $r$ and $\mathtt{X}$. For each $i>0$, define $t_i^{\ast}:= t_i + r(\Theta,\varepsilon',\mathcal{T},|\mathcal{D}|)$. Let $I_i$ be the interval $[t_i,t_i^{\ast}]$, and write $\mathcal{M}[I_i]$ to denote $\mathcal{M}[t_i,t_i^{\ast}]$. Let $\mathcal{U}_1$ be the set of keys controlled by non-faulty processors, and let $\mathcal{U}_2$ be the the set of keys controlled by the adversary.  According to Definition \ref{reflect}, we can then conclude that:

\begin{enumerate} 
\item[$(\$_2)$] It holds with probability $>1-\varepsilon/4$ that, whenever $I_i$ is contained in the duration, there exists some $M\in \mathcal{M}[0,t_i)$ which is broadcast  and an element of   $\mathtt{X}(M,\Theta,\varepsilon',\mathcal{T},|\mathcal{D}|) \cap \mathcal{M}[I_i]$ which $\mathcal{U}_1$ is directed to broadcast, while there does not exist any $M'\in \mathcal{M}[0,t_i)$ which is broadcast and such that $\mathcal{U}_2$ is able to broadcast some element of $\mathtt{X}(M',\Theta,\varepsilon',\mathcal{T},|\mathcal{D}|) \cap \mathcal{M}[I_i]$.
\end{enumerate} 

Since $\mathtt{P}$ is uniformly secure, we also know that: 
\begin{enumerate} 
\item[$(\$_3)$] With probability $> 1-\varepsilon/4$, there do not exist incompatible blocks $B_1,B_2$, timeslots $t_1,t_2$ and $U_1,U_2$ such that $B_i$ is confirmed for $U_i$ at $t_i$ for $i\in \{ 1,2 \}$.
\end{enumerate} 

So now define $\mathtt{X}^{\ast} (\Theta,\varepsilon',\mathcal{T},|\mathcal{D}|)$ to be all those $M=M'\cup M''$ such that $M'\in   \mathtt{X}(M'',\Theta,\varepsilon',\mathcal{T},|\mathcal{D}|)$, and for which there exists $i$ such that all of the following hold: (i)  $I_i \subseteq \mathcal{D}$; (ii) $M'\in \mathcal{M}[I_i]$, $M''\in \mathcal{M}[0,t_i)$ and; (iii) For some chain $C$ of length $i$ with leaf $B$, all messages in $M'$ are attached to $B$ or its descendants. 

Now if  $M\in \mathtt{X}^{\ast} (\Theta,\varepsilon',\mathcal{T},|\mathcal{D}|)$, then let $M',M''$ be such that  $M'\in   \mathtt{X}(M'',\Theta,\varepsilon',\mathcal{T},|\mathcal{D}|)$ and (i)--(iii) above are satisfied, let $i_{M'}$ be the (unique) $i$ such that (i)--(iii) hold w.r.t.\ $M'$, let $C$ be as specified in (iii) for $i_{M'}$, and define $\mathtt{C}^{\ast}(M):= C$. We also define $\mathtt{C}^{\ast}(\emptyset)=\emptyset$. This function $\mathtt{C}^{\ast}$ is almost the notion of confirmation that we want for our recalibration, but the problem is that it is only defined for very specific values of  $M$.  We will use $\mathtt{C}^{\ast}$ to help us define $\mathtt{C}'$ that is defined for all possible $M$. 
Combining $(\$_1)$,  $(\$_2)$ and $(\$_3)$, and the definition of $\mathtt{X}^{\ast}$, it follows that with probability $>1-\varepsilon$ both of the following hold:
\begin{enumerate} 
\item If  $M,M'\in \mathtt{X}^{\ast} (\Theta,\varepsilon',\mathcal{T},|\mathcal{D}|)$ are both broadcast, then all blocks in $\mathtt{C}^{\ast}(M)$ are compatible with all those in $\mathtt{C}^{\ast}(M')$. 
\item For every $i>0$ with $I_i \subseteq \mathcal{D}$,  there exists  $M\in \mathtt{X}^{\ast} (\Theta,\varepsilon',\mathcal{T},|\mathcal{D}|)$ which is broadcast and such that, for some $M',M''$: (i)  $M=M'\cup M''$; (ii) $M'\in   \mathtt{X}(M'',\Theta,\varepsilon',\mathcal{T},|\mathcal{D}|)$; (iii) $M'\in \mathcal{M}[I_i]$, $M''\in \mathcal{M}[0,t_i)$, and; (iv)  For some chain $C$ of length $i$ with leaf $B$, all messages in $M'$ are attached to $B$ or its descendants. 
\end{enumerate} 

In order to define $\mathtt{C}'$ for our recalibration, we can then proceed as follows. Given arbitrary $M$, choose $M'\subseteq M$ such that  $M' \in \mathtt{X}^{\ast} (\Theta,\varepsilon',\mathcal{T},|\mathcal{D}|)$ and $\mathtt{C}^{\ast}(M')$ is of maximal length, or if there exists no $M'$ satisfying these conditions then define $M':=\emptyset$. We define $\mathtt{C}'(M):=\mathtt{C}^{\ast}(M')$. It follows from (1) and (2) above that $(\mathtt{P}',\mathtt{C}')$ produces certificates and satisfies uniform liveness with respect to $\ell'_{\varepsilon,\mathcal{D}}:= \ell_{ \varepsilon',\mathcal{D}}+ 2r(\Theta,\varepsilon',\mathcal{T},|\mathcal{D}|)$.
\end{proof}

\section{Appendix -- Table 1.} 

\begin{table}[H] 
\begin{tabular}{l|l}
term & meaning \\
\hline \hline 
$B $    &                         a block \\
$\mathtt{C}$    &               a notion of confirmation \\ 
$\mathcal{D}$ &          the duration \\ 
$\Delta $     &                    bound on message delay during synchronous \\ 
      & intervals \\ 
 $\varepsilon$ &            the security parameter \\
 $\mathtt{In} $    &             a protocol instance  \\ 
 $m$   &              a message \\
 $M $      &                        a set of messages \\
 $\mathcal{M}$   &         the set of all possible sets of messages \\ 
 $\mathtt{O}$      &               a permitter oracle  \\ 
$p$ & a processor \\
$P$    &                          a permission set \\ 
$\mathtt{P}$     &               a permissionless protocol  \\ 
$R$    &                          a request set  \\ 
$\mathcal{R}$   &          the resource pool \\ 
$\mathtt{S}$     &              a state transition diagram  \\ 
$\sigma$    &     a message \\ 
$t$           &                 a timeslot \\
$(t,\mathtt{U},M,A)$      &  a request in the timed setting  \\ 
$\mathtt{T} $      &               a timing rule \\ 
$\mathtt{U}$ &  a public key \\ 
$(\mathtt{U},M,A)$     &  a request in the untimed setting  \\ 
$\mathcal{U}$ &          the set of all public keys \\ 
$\mathcal{U}_p$ &      the set public keys for $p$ \\

\end{tabular}
\caption{Some commonly used variables and terms.} 
\label{terms} 
\end{table} 

\bibliographystyle{ACM-Reference-Format}

\begin{thebibliography}{18}


\ifx \showCODEN    \undefined \def \showCODEN     #1{\unskip}     \fi
\ifx \showDOI      \undefined \def \showDOI       #1{#1}\fi
\ifx \showISBNx    \undefined \def \showISBNx     #1{\unskip}     \fi
\ifx \showISBNxiii \undefined \def \showISBNxiii  #1{\unskip}     \fi
\ifx \showISSN     \undefined \def \showISSN      #1{\unskip}     \fi
\ifx \showLCCN     \undefined \def \showLCCN      #1{\unskip}     \fi
\ifx \shownote     \undefined \def \shownote      #1{#1}          \fi
\ifx \showarticletitle \undefined \def \showarticletitle #1{#1}   \fi
\ifx \showURL      \undefined \def \showURL       {\relax}        \fi
\providecommand\bibfield[2]{#2}
\providecommand\bibinfo[2]{#2}
\providecommand\natexlab[1]{#1}
\providecommand\showeprint[2][]{arXiv:#2}

\bibitem[\protect\citeauthoryear{Alchieri, Bessani, da~Silva~Fraga, and
  Greve}{Alchieri et~al\mbox{.}}{2008}]%
        {alchieri2008byzantine}
\bibfield{author}{\bibinfo{person}{Eduardo~AP Alchieri},
  \bibinfo{person}{Alysson~Neves Bessani}, \bibinfo{person}{Joni da
  Silva~Fraga}, {and} \bibinfo{person}{Fab{\'\i}ola Greve}.}
  \bibinfo{year}{2008}\natexlab{}.
\newblock \showarticletitle{Byzantine consensus with unknown participants}. In
  \bibinfo{booktitle}{\emph{International Conference On Principles Of
  Distributed Systems}}. Springer, \bibinfo{pages}{22--40}.
\newblock


\bibitem[\protect\citeauthoryear{Bagaria, Kannan, Tse, Fanti, and
  Viswanath}{Bagaria et~al\mbox{.}}{2019}]%
        {bagaria2019prism}
\bibfield{author}{\bibinfo{person}{Vivek Bagaria}, \bibinfo{person}{Sreeram
  Kannan}, \bibinfo{person}{David Tse}, \bibinfo{person}{Giulia Fanti}, {and}
  \bibinfo{person}{Pramod Viswanath}.} \bibinfo{year}{2019}\natexlab{}.
\newblock \showarticletitle{Prism: Deconstructing the blockchain to approach
  physical limits}. In \bibinfo{booktitle}{\emph{Proceedings of the 2019 ACM
  SIGSAC Conference on Computer and Communications Security}}.
  \bibinfo{pages}{585--602}.
\newblock


\bibitem[\protect\citeauthoryear{Bentov, Pass, and Shi}{Bentov
  et~al\mbox{.}}{2016}]%
        {bentov2016snow}
\bibfield{author}{\bibinfo{person}{Iddo Bentov}, \bibinfo{person}{Rafael Pass},
  {and} \bibinfo{person}{Elaine Shi}.} \bibinfo{year}{2016}\natexlab{}.
\newblock \showarticletitle{Snow White: Provably Secure Proofs of Stake.}
\newblock \bibinfo{journal}{\emph{IACR Cryptology ePrint Archive}}
  \bibinfo{volume}{2016}, \bibinfo{number}{919} (\bibinfo{year}{2016}).
\newblock


\bibitem[\protect\citeauthoryear{Buchman}{Buchman}{2016}]%
        {buchman2016tendermint}
\bibfield{author}{\bibinfo{person}{Ethan Buchman}.}
  \bibinfo{year}{2016}\natexlab{}.
\newblock \emph{\bibinfo{title}{Tendermint: Byzantine fault tolerance in the
  age of blockchains}}.
\newblock \bibinfo{thesistype}{Ph.D. Dissertation}.
\newblock


\bibitem[\protect\citeauthoryear{Canetti}{Canetti}{2001}]%
        {canetti2001universally}
\bibfield{author}{\bibinfo{person}{Ran Canetti}.}
  \bibinfo{year}{2001}\natexlab{}.
\newblock \showarticletitle{Universally composable security: A new paradigm for
  cryptographic protocols}. In \bibinfo{booktitle}{\emph{Proceedings 42nd IEEE
  Symposium on Foundations of Computer Science}}. IEEE,
  \bibinfo{pages}{136--145}.
\newblock


\bibitem[\protect\citeauthoryear{Cavin, Sasson, and Schiper}{Cavin
  et~al\mbox{.}}{2004}]%
        {cavin2004consensus}
\bibfield{author}{\bibinfo{person}{David Cavin}, \bibinfo{person}{Yoav Sasson},
  {and} \bibinfo{person}{Andr{\'e} Schiper}.} \bibinfo{year}{2004}\natexlab{}.
\newblock \showarticletitle{Consensus with unknown participants or fundamental
  self-organization}. In \bibinfo{booktitle}{\emph{International Conference on
  Ad-Hoc Networks and Wireless}}. Springer, \bibinfo{pages}{135--148}.
\newblock


\bibitem[\protect\citeauthoryear{Chen, Gorbunov, Micali, and Vlachos}{Chen
  et~al\mbox{.}}{2018}]%
        {chen2018algorand}
\bibfield{author}{\bibinfo{person}{Jing Chen}, \bibinfo{person}{Sergey
  Gorbunov}, \bibinfo{person}{Silvio Micali}, {and} \bibinfo{person}{Georgios
  Vlachos}.} \bibinfo{year}{2018}\natexlab{}.
\newblock \showarticletitle{ALGORAND AGREEMENT: Super Fast and Partition
  Resilient Byzantine Agreement.}
\newblock \bibinfo{journal}{\emph{IACR Cryptol. ePrint Arch.}}
  \bibinfo{volume}{2018} (\bibinfo{year}{2018}), \bibinfo{pages}{377}.
\newblock


\bibitem[\protect\citeauthoryear{Chen and Micali}{Chen and Micali}{2016}]%
        {chen2016algorand}
\bibfield{author}{\bibinfo{person}{Jing Chen} {and} \bibinfo{person}{Silvio
  Micali}.} \bibinfo{year}{2016}\natexlab{}.
\newblock \showarticletitle{Algorand}.
\newblock \bibinfo{journal}{\emph{arXiv preprint arXiv:1607.01341}}
  (\bibinfo{year}{2016}).
\newblock


\bibitem[\protect\citeauthoryear{Dwork, Lynch, and Stockmeyer}{Dwork
  et~al\mbox{.}}{1988}]%
        {DLS88}
\bibfield{author}{\bibinfo{person}{Cynthia Dwork}, \bibinfo{person}{Nancy~A.
  Lynch}, {and} \bibinfo{person}{Larry Stockmeyer}.}
  \bibinfo{year}{1988}\natexlab{}.
\newblock \showarticletitle{Consensus in the Presence of Partial Synchrony}.
\newblock \bibinfo{journal}{\emph{J. ACM}} \bibinfo{volume}{35},
  \bibinfo{number}{2} (\bibinfo{year}{1988}), \bibinfo{pages}{288--323}.
\newblock


\bibitem[\protect\citeauthoryear{Garay, Kiayias, and Leonardos}{Garay
  et~al\mbox{.}}{2018}]%
        {garay2018bitcoin}
\bibfield{author}{\bibinfo{person}{Juan~A Garay}, \bibinfo{person}{Aggelos
  Kiayias}, {and} \bibinfo{person}{Nikos Leonardos}.}
  \bibinfo{year}{2018}\natexlab{}.
\newblock \showarticletitle{The Bitcoin Backbone Protocol: Analysis and
  Applications}.
\newblock  (\bibinfo{year}{2018}).
\newblock


\bibitem[\protect\citeauthoryear{Kiayias, Russell, David, and
  Oliynykov}{Kiayias et~al\mbox{.}}{2017}]%
        {kiayias2017ouroboros}
\bibfield{author}{\bibinfo{person}{Aggelos Kiayias}, \bibinfo{person}{Alexander
  Russell}, \bibinfo{person}{Bernardo David}, {and} \bibinfo{person}{Roman
  Oliynykov}.} \bibinfo{year}{2017}\natexlab{}.
\newblock \showarticletitle{Ouroboros: A provably secure proof-of-stake
  blockchain protocol}. In \bibinfo{booktitle}{\emph{Annual International
  Cryptology Conference}}. Springer, \bibinfo{pages}{357--388}.
\newblock


\bibitem[\protect\citeauthoryear{Lewis-Pye and Roughgarden}{Lewis-Pye and
  Roughgarden}{2021}]%
        {lewis2021byzantine}
\bibfield{author}{\bibinfo{person}{Andrew Lewis-Pye} {and} \bibinfo{person}{Tim
  Roughgarden}.} \bibinfo{year}{2021}\natexlab{}.
\newblock \showarticletitle{Byzantine Generals in the Permissionless Setting}.
\newblock \bibinfo{journal}{\emph{arXiv preprint arXiv:2101.07095}}
  (\bibinfo{year}{2021}).
\newblock


\bibitem[\protect\citeauthoryear{Lynch}{Lynch}{1996}]%
        {lynch1996distributed}
\bibfield{author}{\bibinfo{person}{Nancy~A Lynch}.}
  \bibinfo{year}{1996}\natexlab{}.
\newblock \bibinfo{booktitle}{\emph{Distributed algorithms}}.
\newblock \bibinfo{publisher}{Elsevier}.
\newblock


\bibitem[\protect\citeauthoryear{Nakamoto et~al\mbox{.}}{Nakamoto
  et~al\mbox{.}}{2008}]%
        {nakamoto2008bitcoin}
\bibfield{author}{\bibinfo{person}{Satoshi Nakamoto} {et~al\mbox{.}}}
  \bibinfo{year}{2008}\natexlab{}.
\newblock \bibinfo{title}{Bitcoin: A peer-to-peer electronic cash
  system.(2008)}.
\newblock
\newblock


\bibitem[\protect\citeauthoryear{Okun}{Okun}{2005}]%
        {okun2005distributed}
\bibfield{author}{\bibinfo{person}{Michael Okun}.}
  \bibinfo{year}{2005}\natexlab{}.
\newblock \bibinfo{booktitle}{\emph{Distributed computing among unacquainted
  processors in the presence of Byzantine failures}}.
\newblock \bibinfo{publisher}{Hebrew University of Jerusalem}.
\newblock


\bibitem[\protect\citeauthoryear{Pass, Seeman, and abhi shelat}{Pass
  et~al\mbox{.}}{2016}]%
        {WHGSW16}
\bibfield{author}{\bibinfo{person}{Rafael Pass}, \bibinfo{person}{Lior Seeman},
  {and} \bibinfo{person}{abhi shelat}.} \bibinfo{year}{2016}\natexlab{}.
\newblock \bibinfo{title}{Analysis of the Blockchain Protocol in Asynchronous
  Networks}.
\newblock
\newblock
\newblock
\shownote{eprint.iacr.org/2016/454.}


\bibitem[\protect\citeauthoryear{Ren}{Ren}{2019}]%
        {ren2019analysis}
\bibfield{author}{\bibinfo{person}{Ling Ren}.} \bibinfo{year}{2019}\natexlab{}.
\newblock \bibinfo{booktitle}{\emph{Analysis of nakamoto consensus}}.
\newblock \bibinfo{type}{{T}echnical {R}eport}.
  \bibinfo{institution}{Cryptology ePrint Archive, Report 2019/943.(2019).
  https://eprint. iacr.org}.
\newblock


\bibitem[\protect\citeauthoryear{Yin, Malkhi, Reiter, Gueta, and Abraham}{Yin
  et~al\mbox{.}}{2019}]%
        {yin2019hotstuff}
\bibfield{author}{\bibinfo{person}{Maofan Yin}, \bibinfo{person}{Dahlia
  Malkhi}, \bibinfo{person}{Michael~K Reiter}, \bibinfo{person}{Guy~Golan
  Gueta}, {and} \bibinfo{person}{Ittai Abraham}.}
  \bibinfo{year}{2019}\natexlab{}.
\newblock \showarticletitle{HotStuff: BFT consensus with linearity and
  responsiveness}. In \bibinfo{booktitle}{\emph{Proceedings of the 2019 ACM
  Symposium on Principles of Distributed Computing}}.
  \bibinfo{pages}{347--356}.
\newblock


\end{thebibliography}

\end{document}